%% file: Hardening_book.tex
\documentclass[conference]{IEEEtran}
\IEEEoverridecommandlockouts
\makeatletter
\def\ps@headings{%
	\def\@oddhead{\mbox{}\scriptsize\rightmark \hfil \thepage}%
	\def\@evenhead{\scriptsize\thepage \hfil \leftmark\mbox{}}%
	\def\@oddfoot{}%
	\def\@evenfoot{}}
\makeatother
\usepackage{url}
\usepackage{graphicx}
\usepackage{epstopdf}
\usepackage{algorithm}
\usepackage{algorithmicx}
\usepackage[noend]{algpseudocode}
\usepackage{cite}
\usepackage{color}
\usepackage{amsmath,amssymb,amsfonts,amsthm}
\usepackage{bm}
\usepackage{multirow}
\usepackage{mathtools}
\usepackage{subcaption}
\usepackage{multicol}
\usepackage{lipsum}

\include{notation_20160706}

\newtheorem{theorem} {Theorem}

\newtheorem{lemma} {Lemma}

\newtheorem{corollary} {Corollary}
\newtheorem{rem} {Remark}

\begin{document}
	%
	\title{Enterprise Cyber Resiliency Against Lateral Movement: A Graph Theoretic Approach}

%
%

\author{\IEEEauthorblockN{Pin-Yu~Chen\IEEEauthorrefmark{1},
Sutanay~Choudhury\IEEEauthorrefmark{2},
Luke~Rodriguez\IEEEauthorrefmark{2}, 
Alfred~O.~Hero~III\IEEEauthorrefmark{3} and
Indrajit~Ray\IEEEauthorrefmark{4}}
\IEEEauthorblockA{\IEEEauthorrefmark{1}IBM Research.~ Email: pin-yu.chen@ibm.com}
\IEEEauthorblockA{\IEEEauthorrefmark{2}Pacific Northwest National Laboratory.~
Email: \{Sutanay.Choudhury,Luke.Rodriguez\}@pnnl.gov}
\IEEEauthorblockA{\IEEEauthorrefmark{3}Department of Electrical Engineering and Computer Science, University of Michigan, Ann Arbor.~
	Email: hero@umich.edu}
\IEEEauthorblockA{\IEEEauthorrefmark{4}Department of Computer Science, Colorado State University.~
				Email: indrajit@cs.colostate.edu}	
\thanks{This work was partially supported by the Consortium for Verification Technology under Department of Energy National Nuclear Security Administration award number DE-NA0002534 and by the 
		Asymmetric Resilient Cyber Security initiative at Pacific Northwest National Laboratory, which is operated by Battelle Memorial Institute.}	
			}

	\maketitle

	\begin{abstract}
  Lateral movement attacks are a serious threat to enterprise security. In these attacks, an attacker compromises a trusted user account to get a foothold into the enterprise network and uses it to attack other trusted users, increasingly gaining higher and higher privileges. Such lateral attacks are very hard to model because of the unwitting role that users play in the attack and even harder to detect and prevent because of their low and slow nature. In this paper, a theoretical framework is presented for modeling lateral movement attacks and for proposing a methodology for designing resilient cyber systems against such attacks. The enterprise is modeled as a tripartite graph capturing the interaction between users, machines, and applications, and a set of procedures is proposed to harden the network by increasing the cost of lateral movement. Strong theoretical guarantees on system resilience are established and experimentally validated for large enterprise networks.
	\end{abstract}


	\IEEEpeerreviewmaketitle

\section{Introduction}\label{sec:introduction}

%
%
%
%
		
Cyber security is one of the most critical problems of our
time. Notwithstanding the enormous strides that researchers and
practitioners have made in modeling, analyzing and mitigating cyber
attacks, black hats find newer and newer methods for launching
attacks requiring white hats to revisit the problem with a new
perspective.  One of the major
ways\footnote{http://www.verizonenterprise.com/DBIR} that attackers
launch an attack against an enterprise is by what is known as
``lateral movement via privilege escalation''. This attack cycle,
shown in Fig. \ref{fig:privilege-escalation}, begins with the
compromise of a single user account (not necessarily a privileged
one) in the targeted organization typically via phishing email,
spear phishing or other social engineering techniques.  From this
initial foothold and with time on his side, the attacker begins to
explore the network, possibly compromising other user accounts until
he gains access to a user account with administrative privileges to
the coveted resource: files containing intellectual property,
employee or customer databases or credentials to manage the network
itself.  Typically the attacker compromises multiple intermediate
user accounts, each granting him increasing privileges. Skilled
attackers frequently camouflage their lateral movements into the
normal network traffic making these attacks particularly difficult
to detect and insidious. 

Such lateral attacks are particularly insidious because authorized
users play the role of an unwitting accomplice. End users have often
been recognized as the ``weakest links'' in cyber security \cite{sasse2001transforming}. They do not follow security
advice and often take actions that compromise themselves as well as
others. While efforts to educate and train end users for cyber
security are important steps, anecdotal evidence shows that they have
not been as effective. Clearly, there is a need for designing large
enterprises that are resilient against such lateral movement
attacks. Our current work takes a step in this direction.

\begin{figure}[]
	\centering
	\includegraphics[scale=0.38]{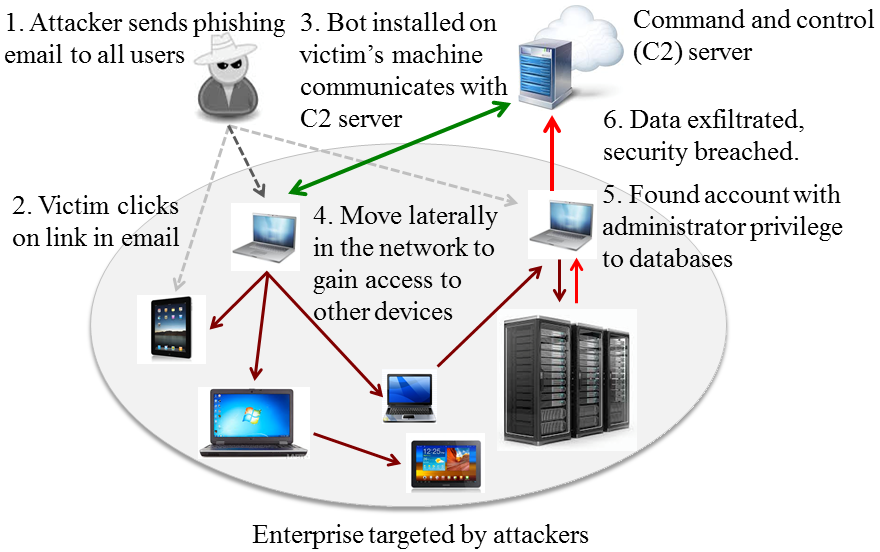}
			\vspace{-4mm}
	\caption{An illustration of a cyber attack using privilege escalation techniques.}
	\label{fig:privilege-escalation}
		\vspace{-6mm}
\end{figure}

\begin{figure*}[t!]
	\centering
	\begin{subfigure}[b]{0.23\textwidth}
		\includegraphics[width=\textwidth]{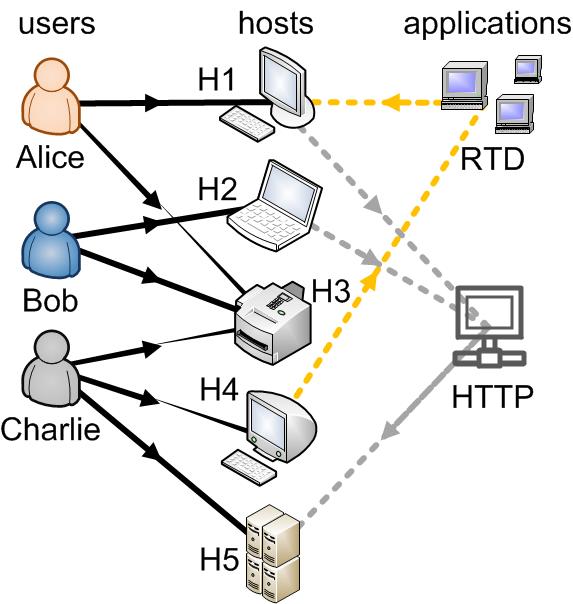}
		\caption{A tripartite network}
		\label{Fig_tripartite_demo}
	\end{subfigure}%
	\hspace{0.1cm}
	\centering
	\begin{subfigure}[b]{0.235\textwidth}
		\includegraphics[width=\textwidth]{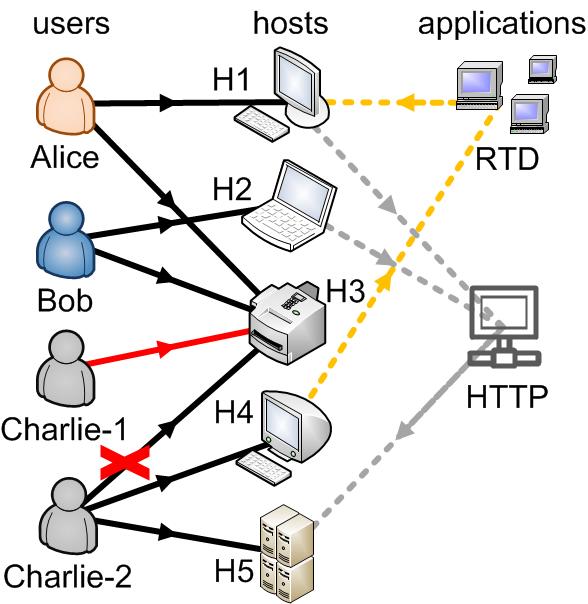}
		\caption{Segmentation}
		\label{Fig_demo_segment}
	\end{subfigure}
	\hspace{0.1cm}
	\centering
	\begin{subfigure}[b]{0.23\textwidth}
		\includegraphics[width=\textwidth]{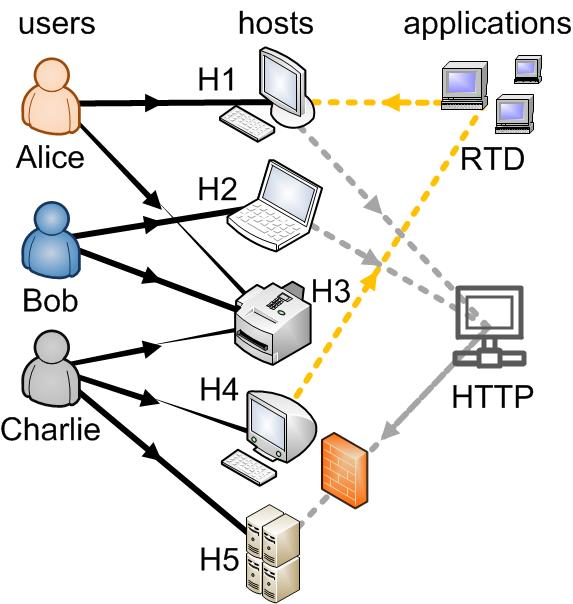}
		\caption{Edge hardening}
		\label{Fig_demo_edge_harden}
	\end{subfigure}
	\hspace{0.1cm}
	\centering
	\begin{subfigure}[b]{0.205\textwidth}
		\includegraphics[width=\textwidth]{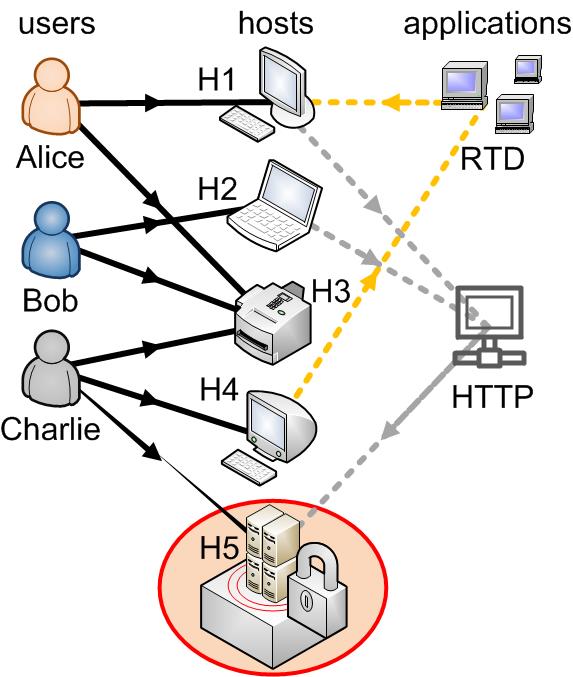}
		\caption{Node hardening}
		\label{Fig_demo_node_harden}
	\end{subfigure}
	\vspace{-4mm}	
	\caption{(a) Illustration of a tripartite network consisting of a set of users, a set of hosts and a set of applications.
		(b) Segmentation - the user Charlie modifies his access configuration by disabling the access of the existing account (Charlie-2) to host H3 
		and by creating a new user account (Charlie-1) for accessing H3 such that an attacker cannot reach the data
		server H5 though the printer H3 if Charlie-2 is compromised.
		(c) Edge hardening  via additional firewall rules on all network flows to H5 through HTTP.
		(d) Node hardening  via system update or security patch installation on H5. 
	}
	\label{Fig_tripartite}
		\vspace*{-6mm}
\end{figure*}


Resilient systems accept that not all attacks can be detected
and prevented; nonetheless, the system should be able to continue
operation even in the face of cyber attacks and provide its core
services or \textsl{missions} even if in a degraded manner
\cite{goldman2011cyber}. To build such a resilient system it is
important to be proactive in understanding and reasoning about
lateral movement in an enterprise network, its potential effects on
the organization, and identify ways to best defend against these
threats. Unfortunately, a theoretical framework for such risk
analysis is currently missing.  Our goal in this paper is to
establish the theoretical foundations of a systematic framework for
building networks resilient to lateral movement attacks.

We model lateral movement attack on an enterprise's mission as a
graph walk on a tripartite {\em user-host-application}
network that logically comprises of two subgraphs: a {\em user-host}
graph and a {\em host-application} graph. Fig. \ref{Fig_tripartite} illustrates the
model and our methodology. The user-host-application paradigm entails richer information than single (homogeneous) graph models (e.g., host-host communication networks). For instance, the host-host communication network can be derived from the host-application subgraph. The user-host-application paradigm also allows us to develop an abstraction of a mission in terms of
concrete entities whose behavior can be monitored and
controlled, which  captures interactions between diverse categories of users,
	software and hardware resources (e.g., virtual machines,
	workstations, mobile devices) and applications.


\begin{table}[!t]
	\centering
	\caption{Utility of the proposed algorithms and established theoretical results.}
	\label{table_unitility_lateral}
	\begin{tabular}{c|c|c}
		\hline
		System Heterogeneity   &  Hardening Methods    & Theoretical Guarantees                      \\ \hline
		User-Host                        & Algorithms \ref{algo_greedy_seg_score}, \ref{algo_greedy_seg_user_host}& Theorem \ref{thm_bound_f}, Corollary \ref{cor_greedy_seg} \\ \hline
		Host-Application             & Algorithms \ref{algo_greedy_edge_harden_score}, \ref{algo_greedy_node_harden} & Theorem \ref{thm_hardening}, Corollary \ref{cor_hardening}\\ \hline
		User-Host-Application          & all of the above      & all of the above                          \\ \hline
	\end{tabular}
		\vspace*{-2mm}
\end{table}

Defining lateral movements as graph walks allows us to determine which nodes in the tripartite graph can
be {\em reached} starting at a given node. From an attacker's perspective,
these nodes that can be ``reached'' are exactly those mission
components that can be attacked and compromised via exploits. The larger the
number of nodes that can be reached by the attacker, the more
``damage'' he/she can cause to the mission. Given a system snapshot and a compromised workstation or
mobile device, we can define the ``Attacker's Reachability'' as a
measure that estimates the number of hosts at risk through a given
number of system exploits. Now, from a defender's
perspective, putting some defensive control on one of these nodes (or edges) allows the walk to be broken at that
point. Intuitively, such a walk can also be used to identify
mission hardening strategies that reduce risk.
This central idea is illustrated in Fig. \ref{Fig_tripartite}.
The heterogeneity of a cyber system entails a network of networks (NoN) representation of entities in the system as displayed in Fig. \ref{Fig_tripartite}, allowing us to devise effective hardening strategies from different perspectives, which differs from works focusing on manipulating the network topology under the assumption that the graph is homogeneous, that is, all nodes have an identical role in a cyber system. 

As our model considers the heterogeneity of a cyber system and incorporates several defensive actions for enhancing the resilience to lateral movement attacks, to assist reading the utility of the proposed approaches and the established theoretical results are summarized in Table \ref{table_unitility_lateral}, and the proofs of the established mathematical results are placed in the appendices in the supplementary file\footnote{Supplementary material: https://goo.gl/h8XHZX}.

The research contributions of this paper are listed as follows.
\begin{enumerate}
	\item By modeling lateral movements as graph walks on a
	user-host-application tripartite graph, we can specify the dominant factors 
	affecting attacker's reachability (Sec. \ref{sec_cascading}), setting the stage for proposing  greedy hardening and segmentation algorithms for network configuration change recommendation to reduce the attacker's reachability  (Sec. \ref{sec_segmentation} and Sec. \ref{sec_hardening}). 
	\item We characterize the effectiveness of  three types of defensive actions against lateral
	movement attacks, each of which can be abstracted via a node or edge operation
	on the tripartite graph, which are 
	(a) segmentation in user-host graph (Sec. \ref{sec_segmentation}),
	(b) edge hardening in host-application graph (Sec. \ref{sec_hardening}), and 
	(c) node hardening in host-application graph (Sec. \ref{sec_hardening}).
	
	\item We provide quantifiable guarantees (e.g.,  submodularity) on the performance loss
	of the proposed greedy algorithms
	relative to the optimal batch
	algorithms that have combinatorial computation complexity (Theorem \ref{thm_bound_f} and Theorem \ref{thm_hardening}). 
	
	\item  We apply our algorithms to a collected
	real tripartite network dataset  and demonstrate that the proposed 
	approaches can significantly constrain attacker's reachability and
	hence provide effective configuration recommendations to secure the
	system (Sec. \ref{sec_expetiment_tripartite}).
	
	\item  We collect traces of real lateral movement attacks in a cyber system  for performance evaluation (Sec. \ref{sec_benchmark}). 
	 We benchmark our approach against the NetMelt algorithm \cite{tong2012gelling_short} and show that our approach can achieve the same reduction in attacker's reachability by hardening nearly 1/3 of the resources as recommended by NetMelt.
\end{enumerate}

\section{Background and Related Work}
Laterally moving through a cyber network looking to obtain access to administrator's credentials or confidential information is a common technique in an attacker's toolbox \cite{provos2003preventing}. Particularly,  privilege escalation through lateral movement 
is a critical challenge for the security community \cite{bugiel2012towards, xing2014upgrading,CPY16_bio_CommMag}. 		
For anomaly detection the authors in \cite{das2010baaz} employ graph clustering to group activities with similar behavioral pattern and make change recommendations
when the access control methods in place deviate from the real-world activity patterns. 
The authors in \cite{chen2012detecting} use community structure to detect anomalous insiders in collaborative information systems. 
 For attack prevention the authors in \cite{zheng2011active} use a graph partitioning approach to fragment the network to limit the possibilities of lateral movement. 
For risk assessment the authors in \cite{Cheng11,CPY14control} use epidemic models for modeling and controlling malware propagation.

Our work fits into two emerging areas of study, 1) \textsl{Network of networks} (NoN) representing multiple inter-related networks as a single model, and 2) studies on resilience of networks. Recently NoN has been an active area of research with diverse topics such as cascading analysis and control in interdependent networks \cite{gao2011robustness, chapman2014controllability}, improved grouping or ranking of entities in a network \cite{ni2014inside}, and mapping of domain problems into the NoN paradigm \cite{halappanavar2013towards}.  Network resilience is a long studied topic \cite{demeester1999resilience}, primarily focusing on the physical topology of communication networks.  There has been a surge in focus on enterprise-level cyber resilience  \cite{goldman2011cyber, choudhury:SafeConfig2015}, where the entire enterprise structure is modeled as a NoN. 

Recently researches have focused on altering the network structure to improve its resilience, as measured in terms of the spectral properties \cite{chan2014make,CPY14ComMag_short}.  Preventing contagion in networks is another attribute for resilience, and approaches such as \cite{prakash2013fractional} suggest algorithms that \textsl{immunize} a subset of nodes as a preventive measure.  
We contribute to this research area by unifying multiple data sources (e.g., different perspectives of user behaviors) into a single model.  Integration of multiple data sources such as user access control and application traffic over the network makes the model more comprehensive and resulting recommendations more profound \cite{hu2014leak}.   
This paper is tailored to providing action recommendations for enhancing the resilience of a heterogeneous cyber system based on the associated NoN representation, which differs from previous works that focus on manipulating the topology of a simple (homogeneous) network where each node in the graph has an identical role \cite{tong2012gelling_short,le2015met_short}.
To the best of our knowledge, this paper proposes the first representation of a cyber system using the NoN model
for designing algorithms that improve resiliency against lateral movement attacks.

\section{Network Model and Iterative Reachability Computation of Lateral Movement}
\label{sec_cascading}

\subsection{Notation and Tripartite Graph Model} 

Throughout this paper a scripted uppercase letter (e.g., $\cX$) denotes a set, a boldfaced uppercase letter (e.g., $\bX$ or $\bX_k$) denotes a matrix, and its entry in the $i$-th row and the $j$-column is denoted by $[\cdot]_{ij}$, a boldfaced lowercase letter (e.g., $\bx$ or $\bx_k$) denotes a column vector, and its $i$-th entry is denoted by $[\cdot]_{i}$, and a plain uppercase or lowercase letter (e.g., $X$ or $x$) denotes a scalar unless specified. 
The expression $|\cX|$ denotes the number of elements in the set $\cX$.
The expression $\mathtt{e}$ denotes the Euler's number, i.e., the base of the natural logarithm.
The expression $\be_i^x$ denotes the $x \times 1$ canonical vector of zero entries except its $i$-th entry is $1$.
The expression $\bI_n$ denotes the $n \times n$ identify matrix. The expression $\bone_n$ denotes the $n \times 1$ column vector of ones.
The expression $\col_x(\bX)$ denotes the $x$-th column of $\bX$.
The expression $\lambda_{\max}(\bX)$ denotes the largest eigenvalue (in magnitude) of a square matrix $\bX$.
The operation $\cdot^T$ denotes matrix or vector transpose.
The operation $\otimes$ denotes the Kronecker product which is defined in Appendix \ref{subsec_Kronecker}.
The operation $\odot$ denotes the Hadamard (entry-wise) product of matrices. 
The operator $\TB: \mathbb{R}_+^n \mapsto [0,1]^n$ is a threshold function such that $[\TB(\bx)]_i=[\bx]_i$ if $0 \leq [\bx]_i \leq 1$ and $[\TB(\bx)]_i=1$ if $[\bx]_i>1$.
The operator $\HB_{\ba}:[0,1]^n \mapsto \{0,1\}^n$ is an entry-wise indicator function such that $[\HB_{\ba}(\bx)]_i=1$ if $[\bx]_i > [\ba]_i $, and  $[\HB_{\ba}(\bx)]_i=0$ otherwise. 
The tripartite graph in Fig. \ref{Fig_tripartite} can be characterized by a set of users 
$\cV_{user}$, a set of hosts $\cV_{host}$, a set of applications $\cV_{app}$, a set of user-host accesses $\cE \subset \cV_{user} \times \cV_{host}$, and a set of host-application-host activities $\cT \subset \cV_{host} \times \cV_{app} \times \cV_{host}$. 
The cardinality of $\cV_{user}$, $\cV_{host}$ and $\cV_{app}$ are denoted by $U$, $N$ and $K$, respectively. The main notation and symbols are listed in Table \ref{table_enterprise_notation}.

\begin{table}[t!]
	\caption{List of main notation and symbols.}		
	\label{table_enterprise_notation}	
	\centering
	\begin{tabular}{|c|l|}
		\hline
		$U / N / K$ & number of users / hosts / applications		
		\\ \hline
		$\lambda_{\max}(\bX)$ & largest eigenvalue of matrix $\bX$
		\\ \hline
		$\otimes$ & Kronecker product
		\\ \hline
		$\bone_n$ & $n \times 1$ column vector of ones
		\\ \hline
		$\bA_C$ & User-host graph matrix
		\\ \hline
		$\bA$ & Host-application graph matrix
		\\ \hline
		$\bP$ & Compromise probability matrix
		\\ \hline
		$\bB$ & $\bA_C^T\bA_C$
		\\ \hline
		$\bJ$ & $\lb \bP \otimes \bone_N \rb^T \bA^T$
		\\ \hline
		$\br$ / $\ba$ & Reachability / Hardening level vector
		\\ \hline
		$\TB(\bx)$ & Threshold function on vector $\bx$
		\\ \hline
		$\HB_{\ba}(\bx)$ & Comparator function of $\bx$ and $\ba$
		\\ \hline
	\end{tabular}	 
	\vspace{-2mm} 
\end{table}

\subsection{Reachability of Lateral Movement on User-Host  Graph}
\label{subsec_reachability_user_host}
Let $G_C=(\cV_{user},\cV_{host},\cE)$ with $\cE \subset \cV_{user} \times \cV_{host}$ denoting the user-host bipartite graph. The access privileges between users and hosts are represented by a binary $U \times N$ adjacency matrix $\bA_C$, where $[\bA_C]_{ij}=1$ if user $i$ can access host $j$, and  $[\bA_C]_{ij}=0$ otherwise.
Let $\br_0$ be an $N \times 1$ binary vector indicating the initial host compromise status, where $[\br_0]_j=1$ if host $j$ is initially being compromised, and  $[\br_0]_j=0$ otherwise. Given $\br_0$, we are interested in computing the final binary compromise vector $\br_\infty$ when attackers leverage user access privileges to compromise other accessible hosts.  The vector $\br_\infty$ specifies the reachability of a lateral movement attack, where reachability is defined as the fraction of hosts that can be reached via graph walks on $G_C$ starting from $\br_0$.
Therefore, reachability is used as a quantitative measure of network vulnerability to lateral movement attacks. Furthermore, studying $\br_\infty$ allows us to investigate the dominant factor that leads to high reachability and more efficient countermeasures.

The computation of $\br_\infty$ can be viewed as a cascading process of repetitive walks on $G_C$ starting from a set of  compromised hosts. Let $\br_t$ denote the binary compromise vector after $t$-hop walks and let $\bw_h$ be the number of $h$-hop walks starting from $\br_0$ and $\bw_0=\br_0$. The hop count of a walk between two hosts in $G_C$ is defined as the number of traversed users.
We begin by computing $\br_1$ from $\br_0$:  the number of 
$1$-hop walk from $\br_0$ to host $j$ is $[\bw_{1}]_j=\sum_{i=1}^U \sum_{k=1}^N [\bA_C]_{ij} [\bA_C]_{ik} [\br_0]_k 
={\be_j^N}^T \bA_C^T\bA_C \br_0.$
Let $\bB=\bA_C^T\bA_C$, an induced adjacency matrix of hosts in $G_C$, 
where  $[\bB]_{ij}$ is the number of common users that can access hosts $i$ and $j$.
Then we have $\bw_1=\bB \br_0$ and $\br_1=\TB (\bw_1)$. Generalizing this result, we have
\begin{align}
\label{eqn_cascade_user_host_1}
\bw_{h+1}&=\bB \bw_h=\bB^{h+1} \br_0;  \\
\label{eqn_cascade_user_host_2}
\br_{t+1}&=\TB \lb \sum_{h=1}^{t+1} \bw_{h} \rb.
\end{align}
The term in (\ref{eqn_cascade_user_host_2}) accounts for the accumulation of compromised hosts up to $t+1$ hops. Note that based on the property of $\TB$, (\ref{eqn_cascade_user_host_2}) can be simplified as 
\begin{align}
\label{eqn_cascade_user_host_3_short}
\br_{t+1}=
\TB \lb \br_t + \bB \br_t \rb. ~~~~\text{(Appendix \ref{proof_eqn_cascade_user_host_3})}
\end{align}
The recursive relation of reachability in (\ref{eqn_cascade_user_host_3_short}) suggests that the term $\bB$ is the dominant factor affecting the propagation of lateral movement. Moreover, from  (\ref{eqn_cascade_user_host_3_short})  we obtain an efficient iterative algorithm for computing $\br_\infty$ that involves successive matrix-vector multiplications until $\br_t$ converges. 



\subsection{Reachability of Lateral Movement on Host-Application  Graph}
\label{subsec_reachability_host_app}
The host-application graph contains the information of host-host communicating through an application. Let $\bA_k$ be an $N \times N$ binary matrix representing the host-to-host communication through application $k$, where 
$[\bA_k]_{ij}=1$ means host $i$ communicates with $j$ through application $k$; and $[\bA_k]_{ij}=0$ otherwise. The $N \times KN$ binary matrix 
$\bA=[\bA_1~\bA_2~\cdots~\bA_K]$ is the concatenated matrix of $K$ host-application-host matrices $\bA_k$ for $k=1,2,\ldots,K$. Let $\bP$ denote the compromise probability matrix, which is a $K \times N$ matrix where its entry $[\bP]_{kj}$ specifies the probability of compromising host $j$ through application $k$. In addition, each host is assigned with a hardening value $[\ba]_j \in [0,1]$ indicating its security level. 

Similar to Sec. \ref{subsec_reachability_user_host}, we are interested in computing the reachability of lateral movement on the host-application  graph. The hop count of a walk between two hosts in the  host-application  graph is defined as the average number of paths between the two hosts through  applications. Let $\bW$ be an $N \times N$ matrix where
$[\bW]_{ij}$ is the average number of one-hop walk from host $i$ to host $j$. Then we have
$[\bW]_{ij}=\sum_{k=1}^K [\bA_k]_{ij} \bP_{kj}$.
Let $\bw_h$ be an $N \times 1$ vector representing the average number of $h$-hop walks of hosts and $\bw_0=\br_0$. 
Then the $j$-th entry of the $1$-hop vector $\bw_1$ is
\begin{align}
\label{eqn_propagation_1_short}
[\bw_1]_j  = \be_j^T \Lb  \col_j(\bP)^T \otimes \bI_n \Rb \bA^T \br_0.~~(\text{Appendix \ref{proof_eqn_propagation_1}})
\end{align} 
Stacking (\ref{eqn_propagation_1_short}) as a column vector gives 
\begin{align}
\label{eqn_propagation_2_short}
\bw_1
=\lb \bP \otimes \bone_N \rb^T \bA^T  \br_0.~~~~(\text{Appendix \ref{proof_eqn_propagation_2}})
\end{align}



The $1$-hop compromise vector $\br_1$ is defined as $\br_1=\HB_{\ba} \lb \TB \lb \bw_1 \rb \rb$. In effect the operator $\HB_{\ba}$ compares the thresholded average number of walks with the hardening level for each host, which means a host $j$ can be compromised only when the thresholded average number of $1$-hop walk $[\TB \lb \bw_1 \rb]_j$ is greater than its hardening level $[\ba]_j$.
Generalizing this result to $h$-hop, we have
\begin{align}
\label{eqn_casecade_host_app_1}
\bw_{h+1}&=\lb \bP \otimes \bone_N \rb^T \bA^T \bw_h;  \\
\label{eqn_casecade_host_app_2}
\br_{t+1}&=\HB_{\ba} \lb \TB \lb \sum_{h=1}^{t+1} \bw_h \rb \rb.
\end{align}
The term in (\ref{eqn_casecade_host_app_2}) has an equivalent expression
\begin{align}
\label{eqn_casecade_host_app_3_short}
\br_{t+1}=\HB_{\ba} \lb \TB \lb \br_t + \lb \bP \otimes \bone_N \rb^T \bA^T \br_t \rb  \rb,
\end{align}
which is proved in Appendix \ref{proof_eqn_casecade_host_app_3}.
Consequently,  for lateral movement on the host-application graph the matrix  $\bJ = \lb \bP \otimes \bone_N \rb^T \bA^T$ is the dominant factor, and (\ref{eqn_casecade_host_app_3_short}) leads to an iterative algorithm for reachability computation.


\subsection{Reachability of Lateral Movement on Tripartite User-Host-Application Graph}

Utilizing the developed results in Sec. \ref{subsec_reachability_user_host} and Sec. \ref{subsec_reachability_host_app},
the cascading process of lateral movement on the tripartite user-host-application graph
can be modeled by
\begin{align}
\br_{t+1}
&\equiv	\HB_{\ba} \lb \TB \lb \br_t +\Lb \bB+ \lb \bP \otimes \bone_N \rb^T \bA^T \br_t \rb  \Rb \rb.  \nonumber
\end{align}


\section{Segmentation on User-Host Graph}
\label{sec_segmentation}
In this section we investigate segmentation on user-host graphs as a countermeasure for suppressing lateral movement. Segmentation works by creating new user accounts to separate user from host in order to reduce the reachability of lateral movement, as illustrated in Fig. \ref{Fig_tripartite} (b). In principle, segmentation removes some edges from the access graph $G_C$ and then merge these removed edges to create new user accounts. Therefore, segmentation retains the same access functionality and constrains lateral movement attacks at the price of additional user accounts. The following analysis provides a theoretical framework of different segmentation strategies.

Recall from (\ref{eqn_cascade_user_host_3_short}) that the matrix $\bB$ is the key factor affecting the reachability of lateral movement on $G_C$. 
Therefore, an effective edge removal approach for segmentation is reducing the spectral radius  of $\bB$ (i.e., $\lambda_{\max}(\bB)$)
by removing some edges from $G_C$. 
Note that by definition $\bB=\bA_C^T \bA_C$ so that $\bB$ is a positive semidefinite (PSD) matrix, and all entries of $\bB$ are nonnegative.
Therefore, by the Perron-Frobenious theorem \cite{HornMatrixAnalysis} the entries of $\bB$'s largest eigenvector $\bu$ (i.e., the eigenvector such that $\bB \bu=\lambda_{\max}(\bB) \bu$) 
are nonnegative.

Here we investigate the change in $\lambda_{\max}(\bB)$ when an edge is removed from $G_C$ in order to define an edge score function that is associated with spectral radius reduction of $\bB$.
If an edge ($i,j) \in \cE$ is removed from $G_C$, then the resulting adjacency matrix of $G_C \setminus (i,j)$ is 
$\bAt_C \lb (i,j) \rb=\bA_C-\be^U_i {\be^N_j}^T$.
The corresponding induced adjacency matrix is 
\begin{align}
\label{eqn_B_one_edge_removal}
\bBt \lb (i,j) \rb &=\bAt_C \lb (i,j) \rb^T \bAt_C \lb (i,j) \rb \nonumber \\
&=\bB-\bA_C^T \be^U_i {\be^N_j}^T - \be^N_j {\be^U_i}^T \bA_C + {\be^N_j} {\be^N_j}^T. 
\end{align}
By the Courant-Fischer theorem \cite{HornMatrixAnalysis} we have 
\begin{align}
\label{eqn_segment_spectral_radius_LB}
\lambda_{\max} \lb \bBt \lb (i,j) \rb \rb  &\geq  \bu^T \bBt \lb (i,j) \rb \bu  \\
&=\lambda_{\max}(\bB) - 2 \bu^T \bA_C^T \be^U_i [\bu]_j + [\bu]_j^2.  \nonumber
\end{align}
The relation in (\ref{eqn_segment_spectral_radius_LB}) leads to a greedy removal strategy that finds the edge $(i,j) \in \cE$ that maximizes the edge score function $ 2 \bu^T \bA_C^T \be^U_i [\bu]_j - [\bu]_j^2$, in order to minimize a lower bound on the spectral radius of $\bBt \lb (i,j) \rb$.
Moreover, Lemma \ref{lem_bound_f} below shows that the edge score function is also associated with an upper bound on the spectral radius of $ \bBt \lb (i,j) \rb$.
Following similar methodology, when a subset of edges $\cE_{\cR} \subset \cE$ are removed from $G_C$, we have 
\begin{align}
\label{eqn_eqn_B_removed_LB_short}
\lambda_{\max} \lb \bBt \lb \cE_{\cR} \rb \rb \geq  
\lambda_{\max}(\bB) - f(\cE_{\cR}),~\text{(Appendix \ref{proof_eqn_B_removed})}
\end{align}
where the function 
\begin{align}
\label{eqn_f_enterprise}
f(\cE_{\cR})&=2 \sum_{(i,j) \in \cE_{\cR}} \bu^T \bA_C^T \be^U_i [\bu]_j \\ \nonumber
&~~~-    \sum_{i \in \cV_{user}} \sum_{j \in \cV_{host}, (i,j) \in \cE_{\cR}} \sum_{s \in \cV_{host}, (i,s) \in \cE_{\cR}}    [\bu]_j [\bu]_s.
\end{align}



In a nutshell, the function $f(\cE_{\cR})$ provides a score that evaluates the effect of edge removal set $\cE_{\cR}$ on the spectral radius of $\bBt \lb \cE_{\cR} \rb$.
The lemma presented in Appendix \ref{proof_lem_f_edge} shows $f(\cE_{\cR})$ is nonnegative as it can be represented as a sum of nonnegative terms.
The following lemma shows that $f(\cE_{\cR})$ is associated with an upper bound on the spectral radius of $\bBt \lb \cE_{\cR} \rb$. Therefore, maximizing $f(\cE_{\cR})$ can
be an effective strategy for spectral radius reduction of $\bB$.
\begin{lemma}
	\label{lem_bound_f}
	For any edge removal set $\cE_{\cR}$ with $|\cE_{\cR}|=q \geq 1$,  if there exits one edge removal set $\cE_{\cR} \subset \cE$ 
	such that $f(\cE_{\cR})>0$, then there exists 
	some constant $c>0$ such that 
	\begin{align}
	&\lambda_{\max}(\bB) - c \cdot f(\cE_{\cR}) \geq	\lambda_{\max} \lb \bBt \lb \cE_{\cR} \rb \rb.  
	\end{align}
\end{lemma}
\begin{proof}
	The proof can be found in Appendix  \ref{proof_lem_bound_f}.
\end{proof}

Moreover, the lemma presented in Appendix \ref{proof_lem_monotone_f} shows that  $f(\cE_{\cR})$ is a monotonic increasing set function, which means that for any two subsets  $\cE_{\cR1},\cE_{\cR2} \subset \cE$ satisfying $\cE_{\cR1} \subset \cE_{\cR2}$,
$f(\cE_{\cR2}) \geq f(\cE_{\cR1})$.
In addition,
the following theorem  shows that  $f(\cE_{\cR})$ is a monotone submodular set function \cite{Fujishige90}, which establishes performance guarantee of greedy edge removal on reducing the spectral radius of $\bBt \lb \cE_{\cR} \rb$.
Submodularity means $f(\cE_{\cR})$ has diminishing gain: for any $\cE_{\cR1} \subset \cE_{\cR2} \subset \cE$ and 
$ e \in \cE \setminus \cE_{\cR2}$,
the discrete derivative $\Delta f(e|\cE_{\cR})=f(\cE_{\cR} \cup e)-f(\cE_{\cR})$ satisfies $\Delta f(e|\cE_{\cR2})\leq \Delta f(e|\cE_{\cR1})$.



\begin{theorem}
	\label{thm_submodularity_edge}
	$f(\cE_{\cR})$ is a monotone submodular set function.
\end{theorem}
\begin{proof}
	The proof can be found in Appendix \ref{proof_thm_submodularity_edge}.
\end{proof}

With the established results, a greedy segmentation algorithm (Algorithm \ref{algo_greedy_seg_score}) is proposed.  Algorithm \ref{algo_greedy_seg_score} computes the edge score function $f \lb (i,j) \rb = 2 \bu^T \bA_C^T \be_i^U [\bu]_j-  [\bu]_j^2$ for every edge $(i,j) \in \cE$  and segments $q$ edges of highest scores to create new user accounts. 
For efficient computation step 2 of Algorithm \ref{algo_greedy_seg_score} can be represented by the matrix form
$\bF = \Lb 2 \bA_C^T \bu \bu^T -\bone_U \lb \bu \odot \bu \rb^T \Rb \odot \bA_C$,
where $[\bF]_{ij}=f \lb (i,j) \rb$ if $(i,j) \in \cE_{\cR}$, and 
$[\bF]_{ij}=0$ otherwise.

\begin{algorithm}[t]
	\caption{Greedy score segmentation algorithm}
	\label{algo_greedy_seg_score}
	\begin{algorithmic}
		\State \textbf{Input:} $\bA_C$, number of segmented edges $q$
		\State \textbf{Output:} modified access adjacency matrix $\bA_C^q$
		\If{recalculating score}
		\State Initialization: $\bA_C^{old} = \bA_C$. $\cE_{old} \leftarrow \cE$. $\cE_{\cR} \leftarrow \varnothing$.
		\For{$z=1$ to $q$}
		\State 1. Compute the leading eigenvector $\bu$ of 
		\State $~~~~\bB={\bA_C^{old}}^T \bA_C^{old}$
		\State 2. Compute score $f \lb (i,j) \rb = 2 \bu^T {\bA_C^{old}}^T \be_i^U [\bu]_j$
		\State~~~~$-  [\bu]_j^2$ for all $(i,j) \in \cE_{old}$
		\State 3. Remove the highest scored edge 
		$(i^*,j^*) \in \cE_{old}$
		\State~~~~from $\bA_C^{old}$  
		\State 4.  $\bA_C^{old} = \bA_C^{old} - \be^U_{i^*} {\be^N_{j^*}}^T $. $\cE_{old} \leftarrow \cE_{old} \setminus (i^*,j^*)$.
		\State~~~~$\cE_{\cR} \leftarrow\cE_{\cR} \cup (i^*,j^*)$.
		\EndFor
		\Else
		\State 1. Compute the leading eigenvector $\bu$ of $\bB=\bA_C^T \bA_C$
		\State 2. Compute score $f \lb (i,j) \rb = 2 \bu^T \bA_C^T \be_i^U [\bu]_j-  [\bu]_j^2$
		\State~~~~for all $(i,j) \in \cE$
		\State 3. Remove the $q$ edges of highest scores from $\bA_C$ 
		\State 4. Store this set of $q$ edges in $\cE_{\cR}$  
		\EndIf
		\State 5. Segment the removed edges in $\cE_{\cR}$ to create new users. A new user $u$ has access to a set of hosts
		$\{s: (u,s) \in \cE_{\cR}\}$
		\State 6. Obtain the modified access adjacency matrix $\bA_C^q$
	\end{algorithmic}
\end{algorithm}

Using the monotonic submodularity of $f(\cE_{\cR})$ in Theorem \ref{thm_submodularity_edge}, the following theorem shows that this greedy algorithm (Algorithm \ref{algo_greedy_seg_score} without score recalculation) has performance guarantee on spectral radius reduction relative to the optimal batch edge removal strategy of combinatorial computation complexity for selecting the best $q$ edges.

\begin{theorem}
	\textnormal{(Greedy segmentation without score recalculation)}~
	\label{thm_bound_f}
	Let $\cE_{\cR}^{opt}$ be the optimal batch edge removal set with $|\cE_{\cR}^{opt}|=q \geq 1$ 
	that maximizes $f(\cE_{\cR})$ and let $\cE_{\cR}^{q}$ with $|\cE_{\cR}^{q}|=q$ be the greedy edge removal set obtained from Algorithm \ref{algo_greedy_seg_score}. If $f(\cE_{\cR}^{q})>0$, then there exists some constant $c^\prime>0$ such that 
	\begin{align}
	&f(\cE_{\cR}^{opt})-f(\cE_{\cR}^{q}) \leq \lb 1-\frac{1}{q} \rb^q f(\cE_{\cR}^{opt}) \leq \frac{1}{\mathtt{e}} f(\cE_{\cR}^{opt}); 
	\nonumber \\
	& f(\cE_{\cR}^{opt}) \geq		\lambda_{\max}(\bB) - \lambda_{\max} \lb \bBt \lb \cE_{\cR}^q \rb \rb \geq  c^\prime  f(\cE_{\cR}^{opt}).  \nonumber	
	\end{align}
\end{theorem}
\begin{proof}
	The proof can be found in Appendix \ref{proof_thm_bound_f}.
\end{proof}


As a variant of Algorithm \ref{algo_greedy_seg_score} without score recalculation, for better traceability one may desire to successively recalculate the largest eigenvector $\bu$ and update the edge score function $f(i,j)$ after each edge removal. 
The following corollary provides a theoretical analysis of the greedy segmentation algorithm with score recalculation (Algorithm \ref{algo_greedy_seg_score} with score recalculation), which shows that score recalculation can successively reduce  the spectral radius of $\bB$.

\begin{corollary}
	\label{cor_greedy_seg}
	\textnormal{(Greedy segmentation with score recalculation)}~
	Let $\bAt(\cE_{\cR})$ denote the adjacency matrix of $G_C \setminus \cE_{\cR}$ for some $\cE_{\cR} \subset \cE$, and let $\bu_{\cE_{\cR}}$ denote the largest eigenvector of $\bBt \lb \cE_{\cR} \rb$.	
	For any edge removal set $\cE_{\cR} \subset \cE$, let $f_{\cE_{\cR}}(i,j)= 2 \bu_{\cE_{\cR}}^T \bAt(\cE_{\cR})^T \be^U_i [\bu_{\cE_{\cR}}]_j - [\bu_{\cE_{\cR}}]_j^2 $, and let $(i^*,j^*)$ be a maximizer of $f_{\cE_{\cR}}(i,j)$.
	Then $\lambda_{\max}\lb \bBt(\cE_{\cR} ) \rb \geq  \lambda_{\max}\lb \bBt(\cE_{\cR} \cup (i^*,j^* )) \rb$. Furthermore, if $f_{\cE_{\cR}}(i^*,j^*)>0$, then $\lambda_{\max}\lb \bBt(\cE_{\cR} ) \rb > \lambda_{\max}\lb \bBt(\cE_{\cR} \cup (i^*,j^* )) \rb$.
\end{corollary}
\begin{proof}
	The proof can be found in Appendix \ref{proof_cor_greedy_seg}.
\end{proof}

In addition to establishing the performance guarantee of greedy score segmentation (Algorithm \ref{algo_greedy_seg_score}) for reducing $\lambda_{\max}(\bB)$, the following theorem shows that the two intuitive  greedy segmentation algorithms proposed in Algorithm \ref{algo_greedy_seg_user_host}, with an aim of successively segmenting the edge connecting to the most connected user or host, are also effectively reducing an upper bound on $\lambda_{\max}(\bB)$. The terms $\bd^U=\bA_C \bone_{N}$ and $\bd^N=\bA_C^T \bone_{U}$ denote the degree vector of users and hosts, respectively, and 
the terms $d_{\max}^{user}$ and $d_{\max}^{host}$ denote the maximum degree of users and hosts in $G_C$, respectively. 
\begin{theorem}
	\textnormal{(Greedy user-(host-)first segmentation)}~
	\label{thm_greedy_user_host}
	If an edge $(i,j)$ is removed from $G_C$ and $\bBt(i,j)$ is irreducible, then
	\begin{align}
	\lambda_{\max} \lb \bBt(i,j) \rb &\leq  d_{\max}^{user} \cdot d_{\max}^{host} \nonumber \\
	&~~~- \max_{s \in \{1,2,\ldots,N\}} \Lb  \lb [\bd^U]_i-1 \rb \be_j^N-\bA_C^T \be_i^U \Rb_s.  	 \nonumber 
	\end{align}
\end{theorem}
\begin{proof}
	The proof and the case when $\bBt(i,j)$ is reducible can be found in Appendix \ref{proof_thm_greedy_user_host}.
\end{proof}

\begin{algorithm}[t]
	\caption{Greedy user-(host-)first segmentation algorithm}
	\label{algo_greedy_seg_user_host}
	\begin{algorithmic}
		\State \textbf{Input:} $\bA_C$, number of segmented edges $q$
		\State \textbf{Output:} modified access adjacency matrix $\bA_C^q$
		\State Initialization: $\bA_C^{old} = \bA_C$. $\cE_{old} \leftarrow \cE$. $\cE_{\cR} \leftarrow \varnothing$.
		\For{$z=1$ to $q$}
		\State 1. Compute user (host) degree vector 
		$\bd^{U}=\bA_C^{old} \bone_{N}$ 	\State~~~~($\bd^{N}={\bA_C^{old}}^T \bone_U$)	
		\State 2. Obtain $i^*=\arg \max_i [\bd^U]_i$ and 
		\State~~~~$j^*=
		\arg \max_{j: [\bA_C^{old}]_{i^*j}>0 } [\bd^{N}]_h$ 
		\State~~~~($j^*=\arg \max_j [\bd^N]_j$ and 
		\State ~~~~$i^*=\arg \max_{i: [\bA_C^{old}]_{ij^*}>0} [\bd^{U}]_i$)	
		\State 3. Remove the  edge $(i^*,j^*) \in \cE_{old}$ from $\bA_C^{old}$.  
		\State 4.  $\bA_C^{old} = \bA_C^{old} - \be^U_{i^*} {\be^N_{j^*}}^T$. $\cE_{old} \leftarrow \cE_{old} \setminus (i^*,j^*)$.
		\State~~~~$\cE_{\cR} \leftarrow \cE_{\cR} \cup (i^*,j^*)$
		\EndFor
		\State 5. Segment the removed edges in $\cE_{\cR}$ to create new 
		users. A new user $u$ has access to a set of hosts
		$\{s: (u,s) \in
		\cE_{\cR}\}$
		\State 6. Obtain the modified access adjacency matrix $\bA_C^{q}$
	\end{algorithmic}
\end{algorithm}

Since the term $\bA_C^T \be_i^U$ in Theorem \ref{thm_greedy_user_host} is a vector of access connections of user $i$,  Theorem \ref{thm_greedy_user_host} suggests a greedy user-first segmentation approach that segments the edge between the user of maximum degree and the corresponding accessible host of maximum degree
in order to reduce the upper bound on spectral radius in  Theorem \ref{thm_greedy_user_host}. Similar analysis apples to the greedy host-first segmentation approach in  Algorithm \ref{algo_greedy_seg_user_host}.

\section{Hardening on Host-Application  Graph}
\label{sec_hardening}
In this section we discuss two countermeasures for constraining lateral movements on the host-application graph. Edge hardening refers to securing access from application $k$ to host $j$, and in effect reducing the compromise probability $[\bP]_{kj}$. Node hardening refers to securing a particular host and in effect increasing its hardening level.

Recall from (\ref{eqn_casecade_host_app_3_short}) that the reachability of lateral movement on host-application graph is governed by the matrix $\bJ=(\bP \otimes \bone_N)^T \bA^T$. Note that although $\bJ$ is in general not a symmetric matrix, its entries are nonnegative and hence by the Perron-Frobenious theorem \cite{HornMatrixAnalysis} $\lambda_{\max}(\bJ)$ is real and nonnegative, and the entries of its largest eigenvector are nonnegative.

Hardening a host $j$ for an application $k$ means that after hardening the compromise probability $[\bP]_{kj}$ is reduced to some value $\epsilon_{kj}$ such that $[\bP]_{kj} > \epsilon_{kj} \geq 0$. Let $\cH$ denote the set of hardened edges  and
let $\bPt_{\cH}$ be the compromise probability matrix after edge hardening.  Then we have
$\bPt_{\cH}=\bP-\sum_{(k,j) \in \cH}  \lb [\bP]_{kj}-\epsilon_{kj}  \rb \be_k^K {\be_j^N}^T.$
Let $\bJt(\cH)=(\bPt_{\cH} \otimes \bone_N)^T \bA^T$ and let $\by$ be the largest eigenvector of $\bJ$. We can show that
\begin{align}
\label{eqn_J_LB_short}
&\lambda_{\max} \lb \bJt (\cH) \rb \geq \lambda_{\max} \lb \bJ  \rb - \by^T \Delta \bJ_{\cH} \by;  \\
\nonumber
&\Delta \bJ_{\cH}=\Lb \lb \sum_{(k,j) \in \cH}
\lb [\bP]_{kj}-\epsilon_{kj}  \rb \be_k^K {\be_j^N}^T \rb \otimes \bone_N   \Rb^T \bA^T.
\end{align}
The proof of  (\ref{eqn_J_LB_short}) can be found in Appendix \ref{proof_eqn_J_LB}.

Let $\phi(\cH)=\by^T \Delta \bJ_{\cH} \by$ be a score function that reflects the effect of the edge hardening set $\cH$ on spectral radius reduction of $\bJ$. 
The lemma presented in Appendix \ref{proof_lem_phi_monotone} shows that $\phi(\cH)$ is a monotonic increasing set function of $\cH$. 
The following analysis shows that
$\phi(\cH)$ is associated with a pair of upper and lower bounds on the spectral radius of $\bJ$ after edge hardening.

\begin{algorithm}[t]
	\caption{Greedy edge hardening algorithm}
	\label{algo_greedy_edge_harden_score}
	\begin{algorithmic}
		\State \textbf{Input:} $\bJ=(\bP \otimes \bone_N)^T \bA^T$, number of hardened edges 
		$\eta$, $\{\epsilon_{kj}\}_{k \in \{1,2,\ldots,K\},j\in \{1,2,\ldots,N\}  }$
		\State \textbf{Output:} modified compromise probability matrix $\bP^\eta$
		\If{recalculating score}
		\State Initialization: $\bP^\eta=\bP$. $\bJ^{old}=\bJ$. 
		\For{$z=1$ to $\eta$}
		\State 1. Compute the leading eigenvector $\by$ of $\bJ^{old}$
		\State 2. Compute score $\phi \lb (k,j) \rb = \by^T \Delta \bJ^{old}_{(k,j)} \by $
		\State 3. Obtain $(k^*,j^*)=\arg \max_{k,j} \phi \lb (k,j) \rb$
		\State 4. Edge hardening: $[\bP^\eta]_{k^* j^*}=\epsilon_{k^* j^*}$
		\State 5. $\bJ^{old}=(\bP^\eta \otimes \bone_N)^T \bA^T$ (see Appendix \ref{proof_efficient_update})
		\EndFor
		\Else
		\State Initialization: $\bP^\eta=\bP$
		\State 1. Compute the leading eigenvector $\by$ of $\bJ$
		\State 2. Compute score $\phi \lb (k,j) \rb = \by^T \Delta \bJ_{(k,j)} \by $
		\State 3. Find the $\eta$ edges of highest scores 
		\State 4. Store this set of $\eta$ edges in $\cH$  
		\State 5. Edge hardening: $[\bP^\eta]_{kj}=\epsilon_{kj}$ for all $(k,j) \in \cH$
		\EndIf
	\end{algorithmic}
\end{algorithm}

The edge hardening algorithm proposed in Algorithm \ref{algo_greedy_edge_harden_score} is a greedy algorithm that hardens the $\eta$ edges of highest scores between applications and hosts, where the per-edge hardening score is defined as $\phi \lb (k,j) \rb = \by^T \Delta \bJ_{(k,j)} \by $. 
Step 5 in Algorithm \ref{algo_greedy_edge_harden_score} with score recalculation can be updated efficiently by tracking the changes in the matrix $\bJ$ caused by Step 4 (see Appendix \ref{proof_efficient_update}).
The following theorem shows that the hardened edge set obtained from Algorithm \ref{algo_greedy_edge_harden_score} without score recalculation is a maximizer of $\phi(\cH)$.
\begin{theorem}
	\label{thm_harden_wo_recal}
	\textnormal{(Greedy edge hardening without score recalculation)}
	For any hardening set $\cH$ with $|\cH|=\eta \geq 1$, let $\cH^{\eta}$ with $|\cH^{\eta}|=\eta$ be the greedy hardening set
	obtained from Algorithm \ref{algo_greedy_edge_harden_score}. Then $\cH^{\eta}$ is a maximizer of $\phi(\cH)$.
\end{theorem}
\begin{proof}
	The proof can be found in Appendix \ref{proof_thm_harden_wo_recal}.
\end{proof}

Furthermore, the following theorem shows that Algorithm \ref{algo_greedy_edge_harden_score} without score recalculation has bounded performance guarantee on spectral radius reduction of $\bJ$ relative to that of the optimal batch edge hardening set for which the computation complexity is combinatorial.
\begin{theorem}
	\textnormal{(Performance guarantee of greedy edge hardening without score recalculation)}
	\label{thm_hardening}
	For any hardening set $\cH$ with $|\cH|=\eta \geq 1$,   $\lambda_{\max} (\bJ) \geq \lambda_{\max} \lb \bJt (\cH) \rb$.  Furthermore, let $\cH^{opt}$ with $|\cH^{opt}|=\eta$ be the optimal hardening set that minimizes $\lambda_{\max} \lb \bJt (\cH) \rb $ and let $\cH^{\eta}$ with $|\cH^{\eta}|=\eta$ be the hardening set that maximizes $\phi(\cH)$.  
	If $\lambda_{\max}(\bJ)>0$ and $\phi(\cH^{\eta})>0$, then there exists some constant $c'' > 0 $ such that
	\begin{align}
	\lambda_{\max}(\bJ)- \phi(\cH^\eta) \leq
	\lambda_{\max} \lb \bJt(\cH^{opt}) \rb \leq
	\lambda_{\max}(\bJ)- c''   \phi(\cH^\eta). 
	\nonumber
	\end{align}
\end{theorem}
\begin{proof}
	The proof can be found in Appendix \ref{proof_thm_hardening}.
\end{proof}

The corollary below
shows Algorithm \ref{algo_greedy_edge_harden_score} with  score recalculation can successively reduce the spectral radius of $\bJ$.
\begin{algorithm}[t]
	\caption{Greedy node hardening algorithm}
	\label{algo_greedy_node_harden}
	\begin{algorithmic}
		\State \textbf{Input:} edge score $\phi((k,j))$, number of hardened 
		nodes $\zeta$, $\{ \alpha_j \}_{j=1}^N$
		\State \textbf{Output:} modified node hardening vector $\bat$
		\State Initialization: $\bat=\ba$
		\State 1. Compute edge hardening score $\phi \lb (k,j) \rb$ for all
			\State~~~~$k \in \{1,2,\ldots,K\}$ and $j \in \{1,2,\ldots,N\}$ 
		\State 2. Compute node hardening score 
		$\rho(j)=\sum_{k=1}^K \phi((k,j))$ 
			\State~~~~for all $j \in \{1,2,\ldots,N\}$ 
		\State 3. Find the first $\zeta$ nodes of highest scores and 
		store this 
			\State~~~~set of $\zeta$ nodes in $\cH^{node}$  
		\State 4. Node hardening: $[\bat]_{j}=\alpha_j$ for all $j \in \cH^{node}$
	\end{algorithmic}
\end{algorithm}

\begin{corollary}
	\label{cor_hardening}
	\textnormal{(Greedy edge hardening with score recalculation)}
	Let $\by_{\cH}$ denote the largest eigenvector of $\bJt(\cH)$ and let $\phi_{\cH} \lb (k,j) \rb=\by_{\cH}^T \bJt \lb \cH \cup (k,j)\rb \by_{\cH}$. For any edge hardening set $\cH$, let $(k^*,j^*)$ be a maximizer of $\phi_{\cH} \lb (k,j) \rb$. Then  $\lambda_{\max} \lb \bJt(\cH) \rb \geq \lambda_{\max} \lb \bJt \lb \cH \cup (k^*,j^*)\rb \rb$. Furthermore, if $\lambda_{\max} \lb \bJt(\cH) \rb>0$ and  $\phi_{\cH} \lb (k^*,j^*) \rb>0$, then  $\lambda_{\max} \lb \bJt(\cH) \rb > \lambda_{\max} \lb \bJt \lb \cH \cup (k^*,j^*)\rb \rb$.
\end{corollary}
\begin{proof}
	The proof can be found in Appendix \ref{proof_cor_hardening}.
\end{proof}

Lastly, for node hardening we use the edge hardening score $\phi((k,j))$ to define the node hardening score $\rho(j)$ for host $j$, where $\rho(j)=\sum_{k=1}^K \phi((k,j)).$
In effect, node hardening on host $j$ enhances its hardening level from $[\ba]_j$ to a value $\alpha_j \in [[\ba]_j,1]$.
A greedy node hardening algorithm based on the node hardening score is summarized in Algorithm \ref{algo_greedy_node_harden}.
In Sec. \ref{sec_expetiment_tripartite} we also investigate the performance of two other node score functions based on $\ba$ and $\bJ$ for greedy node hardening, namely
$\rho^{\ba}(j)=1/[\bfa]_j$ and $\rho^{\bJ}(j)=\sum_{s=1}^N [\bJ]_{js}$.

\section{Experimental Results}
\label{sec_expetiment_tripartite}
\subsection{Dataset Description and Experiment Setup}
To demonstrate the effectiveness of the proposed segmentation and hardening strategies against lateral movement attacks, 
we use the event logs and network flows collected from a large enterprise to create a tripartite user-host-application graph as in Fig. \ref{Fig_tripartite} (a) for performance analysis.
This graph contains 5863 users, 4474 hosts, 3 applications, 8413 user-host access records and 6230  host-application-host network flows. All experiments assume that the defender has no knowledge of which nodes are compromised and the defender only uses the given tripartite network configuration for segmentation and hardening.

To simulate a lateral movement attack
we randomly select 5 hosts (approximates 0.1\% of total host number) as the initially compromised hosts and use the algorithms developed in Sec. \ref{sec_cascading} to evaluate the reachability, which is defined as the fraction of reachable hosts 
by propagating on the tripartite graph from the initially compromised hosts.
The initial node hardening level of each host is independently and uniformly drawn from the unit interval between 0 and 1. The compromise probability matrix $\bP$ is a random matrix where the fraction of nonzero entries is set to be 10\% and each nonzero entry is independently and uniformly drawn from the unit internal between 0 and 1.
The compromise probability after hardening, $\epsilon_{kj}$, is set to be $10^{-5}$ for all $k$ and $j$. All experimental results are averaged over 10 trials.

\begin{figure}[t]
	\centering
	\begin{subfigure}[b]{0.24\textwidth}
		\includegraphics[width=\textwidth]{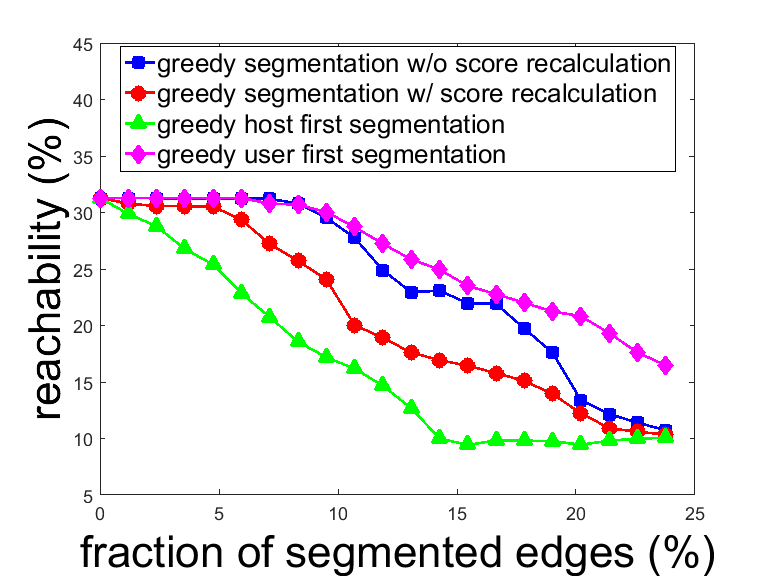}
		\caption{~}
		\label{Fig_segment}
	\end{subfigure}%
	\centering
	\begin{subfigure}[b]{0.24\textwidth}
		\includegraphics[width=\textwidth]{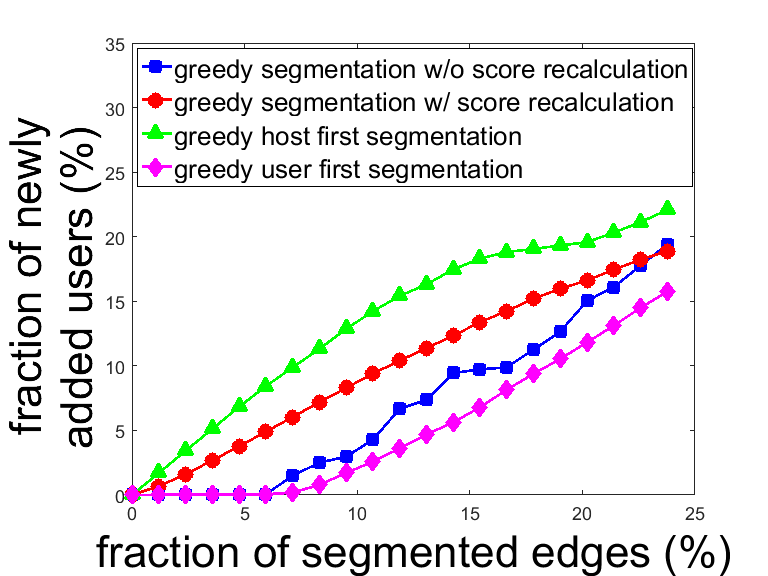}
		\caption{~}
		\label{Fig_segment_new_user}
	\end{subfigure}
	\vspace{-8mm}
	\caption{The effect of segmentation on the user-host access graph. (a) Reachability with respect to different segmentation strategies. (b) Fraction of newly created user accounts from segmentation. Given the same number of segmented edges,  greedy host-first segmentation strategy (green curve) is the most effective approach 
		to constraining reachability (Fig. \ref{Fig_segment_PNNL} (a)) at the cost of most additional accounts (Fig. \ref{Fig_segment_PNNL} (b)).
	}
	\label{Fig_segment_PNNL}
	\vspace{-6mm}
\end{figure}

\subsection{Segmentation against Lateral Movement}
Fig. \ref{Fig_segment_PNNL} shows the effect of different segmentation strategies proposed in Sec. \ref{sec_segmentation} on the user-host graph. In particular, \textit{Fig. \ref{Fig_segment_PNNL} (a) shows that greedy host-first segmentation strategy is the most effective approach to constraining reachability given the same number of segmented edges, since accesses to high-connectivity hosts (i.e., hubs) are segmented. For example, segmenting 15\% of user-host accesses can reduce the reachability to nearly one third of its initial value.} Greedy segmentation with score recalculation is shown to be more effective than that without score recalculation since it is adaptive to user-host access modification during segmentation. Greedy user-first segmentation strategy is not as effective as the other strategies since segmentation does not enforce any user-host access reduction and therefore after segmentation a user can still access the hosts but with different accounts.

Fig.  \ref{Fig_segment_PNNL} (b) shows the fraction of newly created accounts with respect to different segmentation strategies. \textit{There is clearly a trade-off between network security and implementation practicality since Fig.  \ref{Fig_segment_PNNL} suggests that segmentation strategies with better reachability reduction capability also lead to more additional accounts.} 
However, in practice a user might be reluctant to use many accounts to pursue his/her daily jobs even though doing so can greatly mitigate the risk from lateral movement attacks.

\subsection{Hardening against Lateral Movement}
Fig. \ref{Fig_harden_PNNL} shows the effect of different hardening strategies proposed in Sec. \ref{sec_hardening} on the host-application graph. 
As shown in Fig. \ref{Fig_harden_PNNL} (a), the proposed greedy edge hardening strategies with and without score recalculation have similar performance in reachability reduction, and they outperform the greedy heuristic strategy that hardens edges of highest compromise probability. This suggest that the proposed edge hardening strategies indeed finds the nontrivial edges affecting lateral movement.  Fig. \ref{Fig_harden_PNNL} (b) shows that the node hardening strategies using the node score function $\rho$ and $\rho^{\bJ}$ lead to similar performance in reachability reduction, and  they outperform the greedy heuristic strategy that hardens nodes of lowest hardening level.  \textit{These results
	show that the greedy edge and node hardening approaches based on the proposed hardening matrix $\bJ$ outperform heuristics using the compromise probability matrix $\bP$ and the hardening level vector $\ba$, which	suggest that the intuition of hardening the host of lowest security level might not be the best strategy for constraining lateral movement, as it does not take into account the connectivity structure  of the host-application graph.}

\begin{figure}[t]
	\centering
	\begin{subfigure}[b]{0.24\textwidth}
		\includegraphics[width=\textwidth]{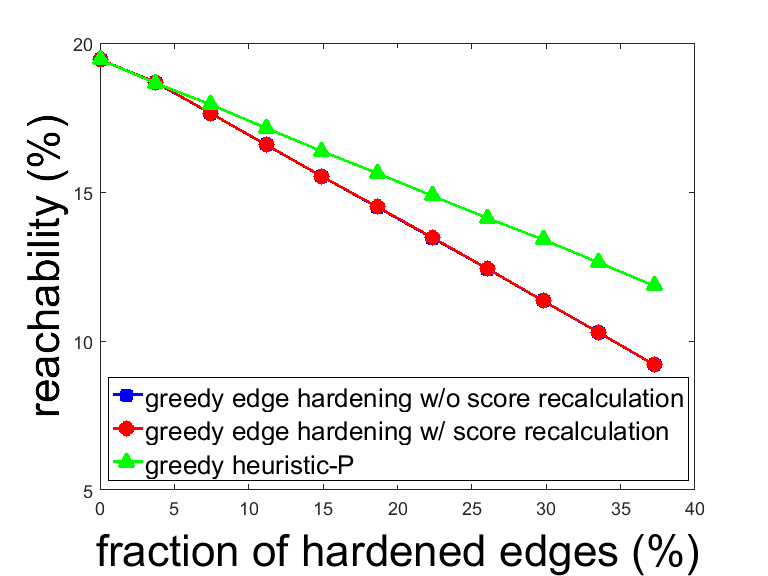}
		\caption{~}
		\label{Fig_edge_harden}
	\end{subfigure}%
	\centering
	\begin{subfigure}[b]{0.24\textwidth}
		\includegraphics[width=\textwidth]{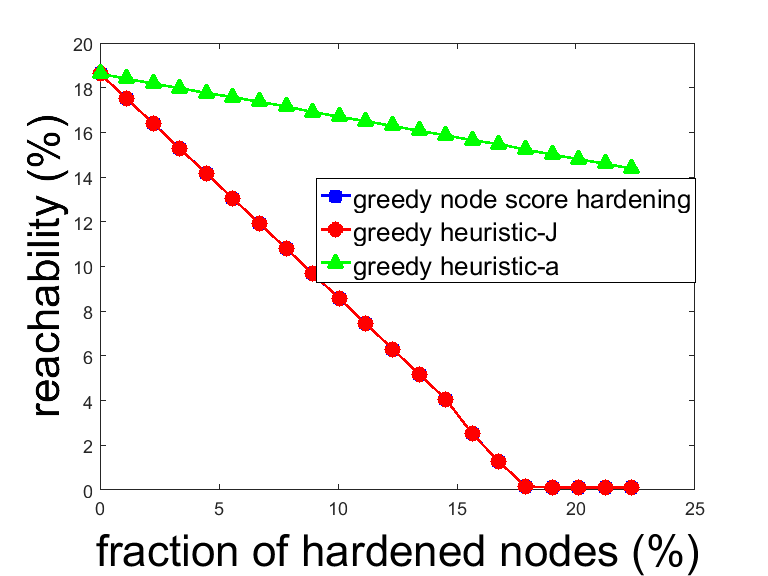}
		\caption{~}
		\label{Fig_node_harden}
	\end{subfigure}
	\vspace{-8mm}
	\caption{The effect of hardening on host-application graph. (a) Reachability with respect to different edge hardening strategies. (b) Reachability with respect to different node hardening strategies. The greedy hardening approaches based on the proposed hardening matrix $\bJ$ (red and blue curves) outperform heuristics using the compromise probability matrix $\bP$ and the hardening level vector $\ba$ (green curve).
		\label{Fig_harden_PNNL}
	}
	\vspace{-6mm}
\end{figure}

\subsection{Segmentation and Hardening on  Tripartite Graph}
Lastly, we investigate the joint effect of segmentation and hardening on constraining lateral movement attacks on the user-host-application tripartite graph. Fig. \ref{Fig_tripartite_PNNL} shows the lateral movement reachability under a selected combination of the proposed segmentation and hardening strategies.
Since these joint segmentation and hardening strategies lead to similar results in reachability reduction, we display their mean and  standard deviation.
In addition, for clarity we only plot representative points to demonstrate the effectiveness. It can be observed that
different combinations of the proposed strategies result in similar tendency in constraining lateral movements.
 Originally, more than half of hosts can be compromised if no preventative actions are taken. Nonetheless, \textit{the proposed segmentation and hardening strategies can greatly reduce the reachability of lateral movements to secure the network.}

\begin{figure}[]
	\centering
	\begin{subfigure}[b]{0.24\textwidth}
		\includegraphics[width=\textwidth]{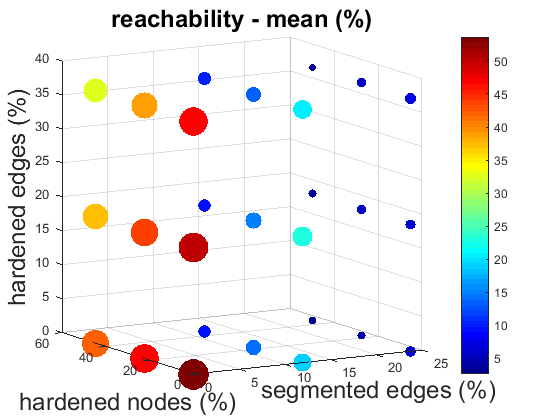}
		\caption{~}
		\label{Fig_tripartite_1_1_1}
	\end{subfigure}%
	\centering
	\begin{subfigure}[b]{0.24\textwidth}
		\includegraphics[width=\textwidth]{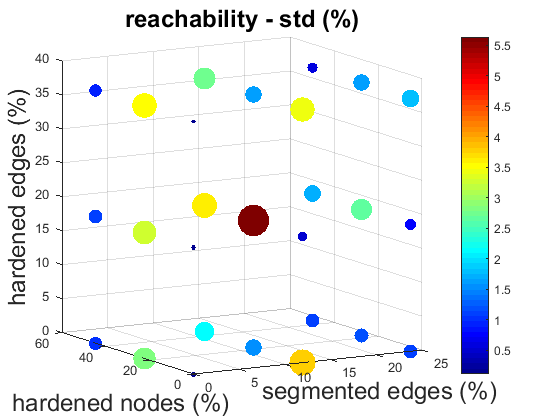}
		\caption{~}
		\label{Fig_tripartite_2_2_1}
	\end{subfigure}%
	\vspace{-4mm}	
	\caption{The effect of segmentation and hardening on lateral movement attack in user-host-application tripartite graph. 	This figure shows the mean and the standard deviation (std) of reachability of four joint segmentation and hardening strategies.
		The size and the color of a point in the plot reflects the level of reachability. 
	}
	\label{Fig_tripartite_PNNL}
	\vspace{-6mm}
\end{figure}

\section{Benchmark: Performance Evaluation on Actual Lateral Movement Attacks}
\label{sec_benchmark}

This section demonstrates the importance of incorporating the heterogeneity of a cyber system for enhancing the resilience to lateral movement attacks. Specifically, real lateral movement attacks taking place in an enterprise network are collected as a performance benchmark\footnote{Dataset available at https://sites.google.com/site/pinyuchenpage/datasets}.
This dataset contains the communication patterns between 2010 hosts via 2 communication protocols, and therefore the enterprise network can be summarized as a bipartite host-application graph. 
It also contains lateral movements originated from a single compromised host, and in total includes 2001 propagation paths. 
The details of the collected benchmark dataset are given in Appendix \ref{proof_bench}. 
The experiment in this section differs from the analysis in Sec. \ref{sec_expetiment_tripartite}, as this dataset contains actual lateral movement traces on the host-application graph, whereas in Sec. \ref{sec_expetiment_tripartite} we have a complete user-host-application tripartite graph of an enterprise, but without the actual attack traces.


We compare the performance of our proposed edge hardening method (Algorithm \ref{algo_greedy_edge_harden_score}) to the \text{NetMelt} algorithm \cite{tong2012gelling_short}, which is a well-known edge removal method for containing information diffusion on a homogeneous graph. For the proposed edge hardening method, the edges in the host-application bipartite graph are hardened sequentially according to the computed scores, and the initial compromise probability matrix $\bP$ is set to be a matrix of ones. For every propagation path, the lateral movement will be contained if the edge it attempts to leverage is hardened. 
Since \text{NetMelt} can only deal with homogeneous graphs (in this case, the host-host graph), its recommendation on hardening a host pair is equivalent to hardening $K$ corresponding host-application edges (in this case, $K=2$), whereas our method has better granularity for edge hardening by considering the connectivity structure of the host-application bipartite graph. The computation  complexity of NetMelt is $O(m \eta+N)$ \cite{tong2012gelling_short}, where $m$ is the number of edges in the host-host graph, $\eta$ is the number of hardened edges, and $N$ is the number of hosts. Since the operation of leading eigenpair computation in Algorithm \ref{algo_greedy_edge_harden_score} is similar to NetMelt, the computation complexity for Algorithm \ref{algo_greedy_edge_harden_score} without score recalculation is $O(m^\prime \eta+N)$, where $m^{\prime}$ is the number of nonzero entries in the matrix $\bJ$. For Algorithm \ref{algo_greedy_edge_harden_score} with score recalculation, the computation complexity is $O(m^\prime \eta^2+N\eta)$.

\begin{figure}[t]
	\centering
	\includegraphics[width=0.28\textwidth]{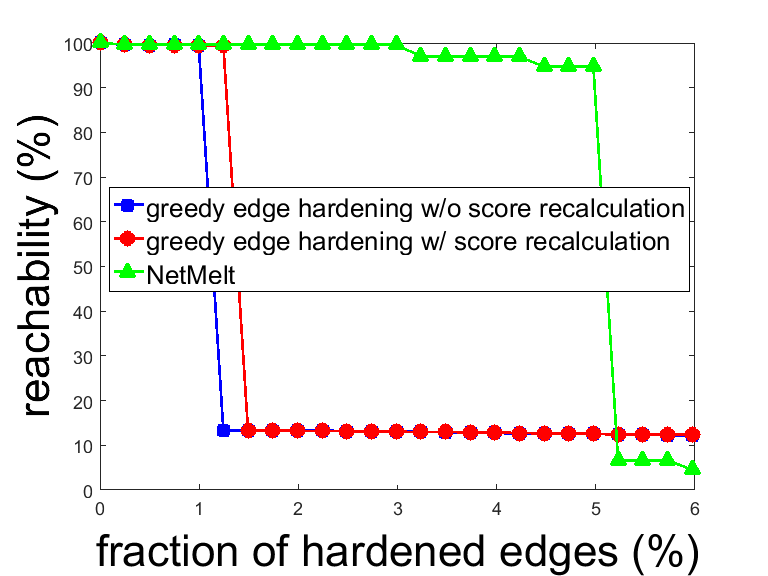}
	\vspace{-1mm}
	\caption{Performance evaluation on the collected benchmark dataset. The proposed approaches (blue and red curves) can restrain the reachability to roughly 10\% by hardening less than 1.5\% of edges, whereas  \text{NetMelt} (green) requires to harden more than 5\% of edges to achieve comparable reachability.} 
	\label{Fig_actual_path}
	\vspace{-6mm}
\end{figure}

Fig. \ref{Fig_actual_path} shows the reachability of lateral movements with respect to the fraction of hardened edges. Initially the reachability is nearly 100\%, suggesting that almost every host is vulnerable to lateral movement attacks without edge hardening. The proposed method (both with or without score recalculation) can restrain the reachability to roughly 10\% by hardening less than 1.5\% of edges, whereas  \text{NetMelt}  requires to harden more than 5\% of edges to achieve comparable reachability, since it does not exploit the heterogeneity of the cyber system. Consequently, the results demonstrate the utility of incorporating heterogeneity for building resilient systems.


\section{Conclusion and Future Work}
This paper developed a framework for joint modeling of multiple dimensions of cyber behavior (user access control, application traffic) for enhancing cyber enterprise resiliency in an unified, tripartite network model.  Our experiments performed on a real dataset demonstrate the value and powerful insights from this unified model with respect to analysis performed on a single dimensional dataset.  Through the tripartite graph model, the dominant factors affecting lateral movement are identified and effective algorithms are proposed to constrain the reachability with theoretical performance guarantees. We also synthesized a benchmark dataset containing traces of actual lateral movement attacks. The results showed that our proposed approach can effectively contain lateral movements by incorporating the heterogeneity of the cyber system. Our future work includes  generalization to $k$-partite networks to model other dimensions of behavior (e.g., authentication mechanisms and social profile of users).

\bibliographystyle{IEEEtran}
\bibliography{IEEEabrv20160824,CPY_ref_20161102_CNS}

\clearpage
\section*{ { \LARGE Supplementary File }}
\appendices
\section{Kronecker Product}
\label{subsec_Kronecker}
If $\bX_1$ is an $r_1 \times \ell_1$ matrix and   $\bX_2$ is an $r_2 \times \ell_2$ matrix, then the Kronecker product $\bX_1 \otimes \bX_2$ is  
an $r_1 r_2 \times \ell_1 \ell_2$ matrix defined as 
\begin{align}
\bX_1 \otimes \bX_2=
\begin{bmatrix}
[\bX]_{11} \bX_2& [\bX]_{12} \bX_2 & \dots  & [\bX]_{1 \ell_1} \bX_2 \\
[\bX]_{21} \bX_2& [\bX]_{22} \bX_2 & \dots  & [\bX]_{2 \ell_1} \bX_2 \\
\vdots & \vdots & \vdots & \vdots \\
[\bX]_{r_1 1} \bX_2& [\bX]_{r_1 2} \bX_2 & \dots  & [\bX]_{r_1 \ell_1} \bX_2
\end{bmatrix}.
\end{align}
Some useful properties of Kronecker product are
\begin{align}
\label{eqn_Kronecker_1}
\lb \bX_1 \otimes \bX_2 \rb^T &= \bX_1^T \otimes \bX_2^T; \\ 
\label{eqn_Kronecker_2}
\bX_1 \otimes \lb \bX_2 + \bX_3 \rb &= \bX_1 \otimes \bX_2 + \bX_1 \otimes \bX_3. 
\end{align}
If $\bX_1$ is an $r_1 \times \ell_1$ matrix, $\bX_2$ is an $r_2 \times \ell_2$ matrix, $\bX_3$ is an $\ell_1 \times \ell_3$ matrix, and   $\bX_4$ is an $\ell_2 \times \ell_4$ matrix,  then
\begin{align}
\label{eqn_Kronecker_3}
(\bX_1 \otimes \bX_2) \cdot (\bX_3 \otimes \bX_4)&=(\bX_1 \cdot \bX_3) \otimes (\bX_2 \cdot \bX_4).
\end{align}

\section{Proof of (\ref{eqn_cascade_user_host_3_short})}
\label{proof_eqn_cascade_user_host_3}
Following (\ref{eqn_cascade_user_host_2}),
\begin{align}
\label{eqn_cascade_user_host_3}
\br_{t+1}&=\TB \lb \sum_{h=1}^{t+1} \bw_{h} \rb  \nonumber \\
&=\TB \lb \sum_{h=1}^{t} \bw_{h} + \bw_{t+1} \rb  \nonumber \\
&\equiv \TB \lb \TB \lb \sum_{h=1}^{t} \bw_{h} \rb + \bw_{t+1} \rb  \nonumber \\
&\equiv	\TB \lb \br_t + \bB \br_t \rb.
\end{align}

\section{Proof of (\ref{eqn_propagation_1_short})}
\label{proof_eqn_propagation_1}
Following the definition of $\bw_1$, we have 
\begin{align}
\label{eqn_propagation_1}
[\bw_1]_j &= \sum_{i=1}^{N} [\br_0]_i [\bW]_{ij}  \nonumber \\
&=\sum_{i=1}^{N} [\br_0]_i \sum_{k=1}^K [\bA_k]_{ij} [\bP]_{kj}  \nonumber \\
&= \sum_{k=1}^K\sum_{i=1}^{N} [\br_0]_i  [\bA_k]_{ij} [\bP]_{kj}  \nonumber \\
&= \sum_{k=1}^K \br_0^T  \bA_k \be_j^N [\bP]_{kj}  \nonumber \\
&=  \br_0^T  \sum_{k=1}^K \lb [\bP]_{kj} \bA_k \rb \be_j^N. 
\end{align} 
Since 
$\sum_{k=1}^K \bP_{kj} \bA_k=\bA \cdot \Lb  \col_j(\bP) \otimes \bI_n \Rb$,
applying it to (\ref{eqn_propagation_1}) we have 
\begin{align}
\label{eqn_eqn_propagation_2}
[\bw_1]_j &= \br_0^T  \bA \Lb  \col_j(\bP) \otimes \bI_n \Rb \be_j \\
&= \be_j^T \Lb  \col_j(\bP)^T \otimes \bI_n \Rb \bA^T \br_0.  \nonumber
\end{align}

\section{Proof of (\ref{eqn_propagation_2_short})}
\label{proof_eqn_propagation_2}
Using   (\ref{eqn_Kronecker_1}) and (\ref{eqn_Kronecker_3}) gives
\begin{align}
\label{eqn_propagation_2}
\bw_1&= \begin{bmatrix}
\be_1^T & \bzero_N^T & \dots    &\bzero_N^T \\
\bzero_N^T & \be_2^T  & \bzero_N^T   &  \vdots \\
\vdots & \vdots & \vdots & \bzero_N^T  \\
\bzero_N^T  &\dots & \bzero_N^T  & \be_N^T
\end{bmatrix}
\cdot
\begin{bmatrix}
\col_1(\bP)^T \otimes \bI_n  \\
\col_2(\bP)^T \otimes \bI_n   \\
\vdots    \\
\col_n(\bP)^T \otimes \bI_n 
\end{bmatrix}
\cdot
\bA^T \br_0  \nonumber \\
&=\lb \bI_n \otimes \bone_N^T \rb \cdot \lb \bP^T \otimes \bI_n \rb \cdot \bA^T \br_0  \nonumber \\
&=\lb \bI_n \cdot \bP^T \rb \otimes \lb \bone_N^T \cdot  \bI_n \rb   \bA^T \br_0  \nonumber \\
&=\lb \bP^T \otimes \bone_N^T \rb \bA^T  \br_0  \nonumber \\
&=\lb \bP \otimes \bone_N \rb^T \bA^T  \br_0. \nonumber \\
\end{align}

\section{Proof of (\ref{eqn_casecade_host_app_3_short})}
\label{proof_eqn_casecade_host_app_3}
Following (\ref{eqn_casecade_host_app_2}),
\begin{align}
\label{eqn_casecade_host_app_3}
\br_{t+1}&= \HB_{\ba} \lb \TB \lb \sum_{h=1}^{t+1} \bw_h \rb \rb \nonumber \\
&=\HB_{\ba} \lb \TB \lb \sum_{h=1}^{t} \bw_h + \bw_{t+1} \rb \rb \nonumber \\
&\equiv 	\HB_{\ba} \lb \TB \lb \TB \lb \sum_{h=1}^{t} \bw_h \rb + \bw_{t+1} \rb  \rb 	\nonumber \\
&\equiv	\HB_{\ba} \lb \TB \lb \br_t + \lb \bP \otimes \bone_N \rb^T \bA^T \br_t \rb  \rb.
\end{align}

\section{Proof of (\ref{eqn_eqn_B_removed_LB_short})}
\label{proof_eqn_B_removed}
When a subset of edges $\cE_{\cR} \subset \cE$ are removed from $G_C$, the resulting adjacency matrix of $G_C \setminus \cE_{\cR}$ is 
\begin{align}
\bAt_C \lb \cE_{\cR} \rb=\bA_C- \sum_{(i,j) \in \cE_{\cR} }  \be^U_i {\be^N_j}^T.
\end{align}
Therefore, the corresponding induced adjacency matrix $\bBt \lb \cE_{\cR} \rb$ is 
\begin{align}
\label{eqn_B_removed}
\bBt \lb \cE_{\cR} \rb&= \bB- \sum_{(i,j) \in \cE_{\cR}}   {\be^N_j} {\be^U_i}^T  \bA_C - \sum_{(i,j) \in \cE_{\cR}} \bA_C^T \be^U_i {\be^N_j}^T \nonumber \\
&~~~+ \sum_{(i,j) \in \cE_{\cR}} \sum_{(\ell,s) \in \cE_{\cR}}   {\be^N_j} {\be^U_i}^T  \be^U_\ell {\be^N_s}^T  \nonumber \\
&=\bB- \sum_{(i,j) \in \cE_{\cR}}   {\be^N_j} {\be^U_i}^T  \bA_C - \sum_{(i,j) \in \cE_{\cR}} \bA_C^T \be^U_i {\be^N_j}^T  \nonumber \\
&~~~+  \sum_{i \in \cV_{user}} \sum_{j \in \cV_{host}, (i,j) \in \cE_{\cR}} \sum_{s \in \cV_{host}, (i,s) \in \cE_{\cR}}   {\be^N_j}  {\be^N_s}^T.
\end{align}

Recall that $\bu$ is the largest eigenvector of $\bB$. Left and right multiplying (\ref{eqn_B_removed}) by $\bu^T$ and $\bu$ and using the Courant-Fischer theorem \cite{HornMatrixAnalysis}, we have 
\begin{align}
\label{eqn_eqn_B_removed_LB_short_repli}
\lambda_{\max} \lb \bBt \lb \cE_{\cR} \rb \rb \geq  
\lambda_{\max}(\bB) - f(\cE_{\cR}),
\end{align}
where $f(\cE_\cR)$ is defined in (\ref{eqn_f_enterprise}).

\section{An equivalent expression of $f(\cE_{\cR})$}
\label{proof_lem_f_edge}
The following lemma provides an equivalent representation of the function $f(\cE_{\cR})$ in (\ref{eqn_eqn_B_removed_LB_short}), which also implies that $f(\cE_{\cR})$ is nonnegative as it can be represented by a sum of nonnegative terms.
\begin{lemma}
	\label{lem_f_edge}
	Let $\varnothing$ denote the empty set. Then $f(\varnothing) = 0$ and 
	\begin{align}
	f(\cE_{\cR})&=\sum_{i \in \cV_{user}} \sum_{j \in \cV_{host}, (i,j) \in \cE_{\cR}} \sum_{s \in \cV_{host}, (i,s) \in \cE_{\cR}}   [\bu]_j [\bu]_s \nonumber \\
	&~~~+2	 \sum_{i \in \cV_{user}} \sum_{j \in \cV_{host},(i,j) \in \cE_{\cR}} \sum_{s \in \cV_{host},(i,s) \in \cE / \cE_{\cR}}   [\bu]_j [\bu]_s.	\nonumber
	\end{align}
\end{lemma}
\begin{proof}
	$f(\varnothing) = 0$ is a direct result from the definition of $f(\cE_{\cR})$. Moreover, $f(\cE_{\cR})$
	has an equivalent expression that
	\begin{align}
	f(\cE_{\cR})&=2 \sum_{(i,j) \in \cE_{\cR}} \bu^T \bA_C^T \be^U_i [\bu]_j  \nonumber \\
	&~~~-   \sum_{i \in \cV_{user}} \sum_{ \in \cV_{host},(i,j) \in \cE_{\cR}} \sum_{s \in \cV_{host},(i,s) \in \cE_{\cR}}   [\bu]_j [\bu]_s  \nonumber \\
	&=  2 \sum_{(i,j) \in \cE_{\cR}}  \sum_{s \in \cV_{host}} [\bA_C]_{i s} [\bu]_s  [\bu]_j \nonumber \\   
	&~~~-   \sum_{i \in \cV_{user}} \sum_{ j \in \cV_{host},(i,j) \in \cE_{\cR}} \sum_{s \in \cV_{host},(i,s) \in \cE_{\cR}}   [\bu]_j [\bu]_s  \nonumber \\
	&= 2 \sum_{i \in \cV_{user}}  \sum_{j \in \cV_{host},(i,j) \in \cE_{\cR}}  \lb \sum_{s \in \cV_{host},(i,s) \in \cE_{\cR}} \right. \nonumber \\
	&~~~\left.+ \sum_{s \in \cV_{host},(i,s) \in \cE/ \cE_{\cR}}  \rb [\bu]_j  [\bu]_s  \nonumber \\
	&~~~-   \sum_{i \in \cV} \sum_{ \in \cV,(i,j) \in \cE_{\cR}} \sum_{s \in \cV,(i,s) \in \cE_{\cR}}   [\bu]_j [\bu]_s    \nonumber \\
	&=\sum_{i \in \cV} \sum_{j \in \cV,(i,j) \in \cE_{\cR}} \sum_{s \in \cV,(i,s) \in \cE_{\cR}}   [\bu]_j [\bu]_s \nonumber \\
	&~~~+2	 \sum_{i \in \cV} \sum_{j \in \cV,(i,j) \in \cE_{\cR}} \sum_{s \in \cV,(i,s) \in \cE / \cE_{\cR}}   [\bu]_j [\bu]_s.	
	\end{align}
	The nonnegativity  of $\bu$ suggests that $f(\cE_{\cR})\geq 0$.
\end{proof}

\section{Proof of Lemma \ref{lem_bound_f}}
\label{proof_lem_bound_f}
For any edge removal set $\cE_{\cR} \subset \cE$ with $|\cE_{\cR}|=q$, let $\bv$ be the largest eigenvector of $\bBt \lb \cE_{\cR} \rb$.
Left and right multiplying (\ref{eqn_B_removed}) by $\bv^T$ and $\bv$ gives
\begin{align}
\label{eqn_bound_f_1}
\lambda_{\max} \lb \bBt \lb \cE_{\cR} \rb \rb & =	\bv^T \bB \bv - g(\cE_{\cR})  \nonumber \\
& \leq 		\lambda_{\max}(\bB) - g(\cE_{\cR})
\end{align}
by the Courant-Fischer theorem \cite{HornMatrixAnalysis}, where 
$g(\cE_{\cR})=\sum_{i \in \cV_{user}} \sum_{j \in \cV_{host},(i,j) \in \cE_{\cR}} \sum_{s \in \cV_{host},(i,s) \in \cE_{\cR}}   [\bv]_j [\bv]_s$\\
$+2	 \sum_{i \in \cV_{user}} \sum_{j \in \cV_{host},(i,j) \in \cE_{\cR}} \sum_{s \in \cV_{host},(i,s) \in \cE / \cE_{\cR}}   [\bv]_j [\bv]_s$
is obtained by following the same derivation procedure as in Lemma \ref{lem_f_edge}. 

Next, recall from the Perron-Frobenius theorem \cite{HornMatrixAnalysis} that the entries of $\bu$ and $\bv$ are all nonnegative and bounded. 
Therefore, there must exist one edge removal set $\cE_{\cR}$ with $|\cE_{\cR}|=q$ such that $g(\cE_{\cR})>0$. Otherwise  $g(\cE_{\cR})=0$ for every edge removal set with cardinality $|\cE_{\cR}|=q \geq 1$ implies that $\bv$ is a zero vector, which contradicts the fact that $\bv$ is an eigenvector.
Finally, since $f(\cE_{\cR})>0$,
there exists a constant $c>0$ such that $g(\cE_{\cR}) \geq c \cdot f(\cE_{\cR})$. Applying this inequality to (\ref{eqn_bound_f_1}) gives $	\lambda_{\max}  \lb \bBt \lb \cE_{\cR} \rb \rb  \leq 	\lambda_{\max}(\bB) - c \cdot f(\cE_{\cR})$. 

\section{Monotonicity of $f(\cE_{\cR})$}
\label{proof_lem_monotone_f}
\begin{lemma}
	\label{lem_monotone_f}
	$f(\cE_{\cR})$ is a monotonic increasing set function.
\end{lemma}
\begin{proof}
	For any two subsets  $\cE_{\cR1},\cE_{\cR2} \subset \cE$ satisfying $\cE_{\cR1} \subset \cE_{\cR2}$, let $\Delta \cE_{\cR}= \cE_{\cR2} / \cE_{\cR1}$. 	
	Using the relation $ \cE_{\cR2}=\cE_{\cR1} \cup \Delta \cE_{\cR}$ and $\cE_{\cR1} \cap \Delta \cE_{\cR} = \varnothing$,
	from Lemma \ref{lem_f_edge} 	$f(\cE_{\cR2})$ can be represented as
	\begin{align}
	f(\cE_{\cR2})&=\sum_{i \in \cV_{user}} \lb \sum_{j \in \cV_{host},(i,j) \in \cE_{\cR1}} + \sum_{j \in \cV_{host},(i,j) \in \Delta \cE_{\cR}} \rb \nonumber \\
	&~~~\cdot
	 \lb \sum_{s \in \cV_{host},(i,s) \in \cE_{\cR1}} + \sum_{s \in \cV_{host},(i,s) \in \Delta \cE_{\cR} }  \rb  [\bu]_j [\bu]_s  \nonumber \\
	&~+2	 \sum_{i \in \cV_{user}} \lb \sum_{j \in \cV_{host},(i,j) \in \cE_{\cR1}} + \sum_{j \in \cV_{host},(i,j) \in \Delta \cE_{\cR}} \rb \nonumber \\
	&~~~\cdot \sum_{s \in \cV_{host},(i,s) \in \cE \setminus \cE_{\cR}}   [\bu]_j [\bu]_s. 	
	\end{align}
	Similarly, using the relation $ \Delta \cE_{\cR} = (\cE \setminus \cE_{\cR1}) \setminus (\cE \setminus \cE_{\cR2})$, from Lemma \ref{lem_f_edge} we have
	\begin{align}
	f(\cE_{\cR1})&=\sum_{i \in \cV_{user}} \sum_{j \in \cV_{host},(i,j) \in \cE_{\cR1}} \sum_{s \in \cV_{host},(i,s) \in \cE_{\cR1}}   [\bu]_j [\bu]_s  \nonumber \\
	&~+2 \sum_{i \in \cV_{user}} \sum_{j \in \cV_{host},(i,j) \in \cE_{\cR1}} \nonumber \\
	&~~~\cdot \lb \sum_{s \in \cV_{host},(i,s) \in \cE \setminus \cE_{\cR2}} + \sum_{s \in \cV_{host},(i,s) \in \cE \setminus \Delta \cE_{\cR}}\rb 
	[\bu]_j [\bu]_s.  
	\end{align}
	Therefore,
	\begin{align}
	&f(\cE_{\cR2})-	f(\cE_{\cR1}) \nonumber \\ 
	&= \sum_{i \in \cV_{user}}  \sum_{j \in \cV_{host},(i,j) \in \Delta \cE_{\cR}} 
	\sum_{s \in \cV_{host},(i,s) \in \cE_{\cR2} }    [\bu]_j [\bu]_s   \nonumber \\
	&~+2	 \sum_{i_{user} \in \cV} \sum_{j \in \cV_{host},(i,j) \in \Delta \cE_{\cR}}  \sum_{s \in \cV_{host},(i,s) \in \cE \setminus  \cE_{\cR2}}  [\bu]_j [\bu]_s  \nonumber \\
	\label{eqn_monotone_f_1}
	&~-	 \sum_{i \in \cV_{user}} \sum_{j \in \cV_{host},(i,j) \in \cE_{\cR1}}  \sum_{s \in \cV_{host},(i,s) \in \Delta \cE_{\cR}}  [\bu]_j [\bu]_s 
	\nonumber \\
	&\geq \sum_{i \in \cV_{user}}  \sum_{j \in \cV_{host},(i,j) \in \Delta \cE_{\cR}} 
	\sum_{s \in \cV_{host},(i,s) \in \cE_{\cR1} }    [\bu]_j [\bu]_s  \\
	&~+2	 \sum_{i \in \cV_{user}} \sum_{j \in \cV_{host},(i,j) \in \Delta \cE_{\cR}}  \sum_{s \in \cV_{host},(i,s) \in \cE \setminus  \cE_{\cR2}}  [\bu]_j [\bu]_s  \nonumber \\
	&~-	 \sum_{i \in \cV_{user}} \sum_{j \in \cV_{host},(i,j) \in \cE_{\cR1}}  \sum_{s \in \cV_{host},(i,s) \in \Delta \cE_{\cR}}  [\bu]_j [\bu]_s  
	\nonumber \\ 					
	&=2	 \sum_{i \in \cV_{user}} \sum_{j \in \cV_{host},(i,j) \in \Delta \cE_{\cR}}  \sum_{s \in \cV_{host},(i,s) \in \cE \setminus  \cE_{\cR2}}  [\bu]_j [\bu]_s  
	\\
	\label{eqn_monotone_f_4}
	& \geq 0,		
	\end{align}
	where the inequality in (\ref{eqn_monotone_f_1}) uses the Perron-Frobenious theorem \cite{HornMatrixAnalysis} that $[\bu]_s \geq 0$ and the fact that $\cE_{\cR1} \subset \cE_{\cR2}$. 	The inequality in (\ref{eqn_monotone_f_4}) is due to the nonnegativity of the largest eigenvector $\bu$.
\end{proof}
\section{Proof of Theorem \ref{thm_submodularity_edge}}
\label{proof_thm_submodularity_edge}
It has been proved in Lemma \ref{lem_monotone_f} that 	$f(\cE_{\cR})$ is a monotone increasing set function. Here we prove that $f(\cE_{\cR})$ is submodular. 
For any $\cE_{\cR1} \subset \cE_{\cR2} \subset \cE$ and 
$ e \in \cE \setminus \cE_{\cR2}$,	let $e=(u,v) \in \cE$, from (\ref{eqn_monotone_f_1}) we have 
\begin{align}
&\Delta f(e|\cE_{\cR2}) = f(\cE_{\cR2} \cup e)-f(\cE_{\cR2})  \nonumber \\
&= \sum_{i \in \cV_{user}}  \sum_{j \in \cV_{user},(i,j) = e} 
\lb \sum_{s \in \cV_{host},(i,s) \in \cE_{\cR2} }   + \sum_{s \in \cV_{host},(i,s)=e }\rb \nonumber \\
&~~~\cdot  [\bu]_j [\bu]_s   \nonumber \\
&~+2	 \sum_{i \in \cV_{user}} \sum_{j \in \cV_{host},(i,j) = e}  \sum_{s \in \cV_{host},(i,s) \in \cE \setminus ( \cE_{\cR2} \cup e )}  [\bu]_j [\bu]_s  \nonumber \\
&~-	 \sum_{i \in \cV_{user}} \sum_{j \in \cV_{host},(i,j) \in \cE_{\cR2}}  \sum_{s \in \cV_{host},(i,s) =e}  [\bu]_j [\bu]_s  \nonumber \\
&=[\bu]_u [\bu]_v + 2  \sum_{s \in \cV_{host},(u,s) \in \cE \setminus ( \cE_{\cR2} \cup e )}  [\bu]_u [\bu]_s \nonumber \\
& \leq [\bu]_u [\bu]_v + 2  \sum_{s \in \cV_{host},(u,s) \in \cE \setminus ( \cE_{\cR1} \cup e )}  [\bu]_u [\bu]_s \\
\label{eqn_submodular_f_3}
&=  \Delta f(e|\cE_{\cR1}),
\end{align}
where the inequality in (\ref{eqn_submodular_f_3}) holds due to the fact that $\cE \setminus ( \cE_{\cR2} \cup e ) \subset \cE \setminus ( \cE_{\cR1} \cup e )$ and the entries of $\bu$ are nonnegative from the Perron-Frobenious theorem \cite{HornMatrixAnalysis}. Therefore, $f(\cE_{\cR})$ is a monotone submodular set function.


\section{Proof of Theorem \ref{thm_bound_f}}
\label{proof_thm_bound_f}
Let $\cE_{\cR}^s$ with $|\cE_{\cR}^s|=s$ be the greedy edge removal set obtained from Algorithm \ref{algo_greedy_seg_score}. 
By submodularity of $f(\cE_{\cR})$ from  Theorem \ref{thm_submodularity_edge}, for every $s < q$ there exists an edge $e \in \cE_{\cR}^{opt} / \cE_{\cR}^{s}$  such that
\begin{align}
f(\cE_{\cR}^{s} \cup e)-f(\cE_{\cR}^{s}) \geq \frac{1}{q} \left( f(\cE_{\cR}^{opt})-f(\cE_{\cR}^{s}) \right).
\end{align}
After algebraic manipulation, we have
\begin{align}
\label{eqn_thm_bound_f_1}
f(\cE_{\cR}^{opt})-f(\cE_{\cR}^{s+1}) \leq \left(1-\frac{1}{q}\right) \lb f(\cE_{\cR}^{opt})-f(\cE_{\cR}^{s}) \rb
\end{align}
and therefore by telescoping (\ref{eqn_thm_bound_f_1}) we have
\begin{align}
\label{eqn_thm_bound_f_2}
f(\cE_{\cR}^{opt})-	f(\cE_{\cR}^{q}) \leq \left(1-\frac{1}{q}\right)^q 	f(\cE_{\cR}^{opt}) \leq \frac{1}{\mathtt{e}} 	f(\cE_{\cR}^{opt}).
\end{align}
Applying (\ref{eqn_thm_bound_f_2}) and the fact that $0<f(\cE_{\cR}^{q}) \leq f(\cE_{\cR}^{opt})$ to (\ref{eqn_eqn_B_removed_LB_short}), there exists some constant $c>0$ such that
\begin{align}
&\lambda_{\max}(\bB) - c  \lb 1- \mathtt{e}^{-1}\rb \cdot f(\cE_{\cR}^{opt}) \geq	\lambda_{\max} \lb \bBt \lb \cE_{\cR}^q \rb \rb; \\
&\lambda_{\max} \lb \bBt \lb \cE_{\cR}^q \rb \rb \geq  	\lambda_{\max}(\bB) - f(\cE_{\cR}^{opt}).   
\end{align}
The proof is complete by setting $c^\prime=c  \lb 1- \mathtt{e}^{-1}\rb$.

\section{Proof of Corollary \ref{cor_greedy_seg}}
\label{proof_cor_greedy_seg}
%
This corollary is a direct result of Lemma \ref{lem_bound_f} and Theorem \ref{thm_bound_f} by replacing $\bB$ with $\bBt \lb \cE_{\cR} \rb$ and setting $q=1$.

\section{Proof of Theorem \ref{thm_greedy_user_host}}
\label{proof_thm_greedy_user_host}
We use the fact from the Perron-Frobenius theorem that if a square matrix $\bX$ is irreducible and nonnegative, then 
$ \lambda_{\max}(\bX) \leq \max_{s} \sum_t [\bX]_{st}$. A square nonnegative matrix $\bX$ is irreducible means that 
for every pair of indices $s$ and $t$, 
there exists a natural number $z$ such that $[\bX^z]_{st}>0$.   
Since $\bBt(i,j)$ is a matrix of nonnegative entries, if $\bBt(i,j)$ is irreducible, from (\ref{eqn_B_one_edge_removal}) we have
\begin{align}
\lambda_{\max} \lb \bBt(i,j) \rb &\leq \max_{s \in \{1,2,\ldots,N\}} \Lb \bBt(i,j) \bone_N \Rb_s   \nonumber \\		
&=\max_{s \in \{1,2,\ldots,N\}} \Lb \bB \bone_N -\bA_C^T \be^U_i - [\bd^U]_i \be^N_j  + {\be^N_j}  \Rb_s \nonumber 	 \\
\label{eqn_thm_greedy_user_host_3}
&\leq d_{\max}^{user} \cdot d_{\max}^{host} \nonumber \\ 
&~~~- \max_{s \in \{1,2,\ldots,N\}} \Lb  \lb [\bd^U]_i-1 \rb \be_j^N-\bA_C^T \be_i^U \Rb_s,  	 
\end{align}
where (\ref{eqn_thm_greedy_user_host_3}) uses the fact that for all $t \in \{1,2,\ldots,N\}$,
\begin{align}
[\bB \bone_N]_t=[\bA_C^T \bA_C \bone_N]_t =	[\bA_C^T \bd_U]_t \leq 	 d_{\max}^{user} \cdot d_{\max}^{host}. 
\end{align}
\begin{rem}
	If $\bBt(i,j)$ is reducible, one can obtain a similar upper bound as in  Theorem \ref{thm_greedy_user_host} since the largest eigenvalue of $\bBt(i,j)$ is the maximum value of the largest eigenvalue of block-wise irreducible nonnegative submatrices of $\bBt(i,j)$. 
\end{rem}

\section{Proof of (\ref{eqn_J_LB_short})}
\label{proof_eqn_J_LB}
By the Courant-Fischer theorem \cite{HornMatrixAnalysis}, (\ref{eqn_Kronecker_1}) and (\ref{eqn_Kronecker_2}) we have
\begin{align}
\label{eqn_J_LB}
\lambda_{\max} \lb \bJt (\cH) \rb &\geq \by^T \bJt(\cH) \by  \nonumber \\
&=\by^T (\bPt_{\cH} \otimes \bone_N)^T \bA^T \by \nonumber \\
&=\lambda_{\max} \lb \bJ  \rb - \by^T \Delta \bJ_{\cH} \by,
\end{align}
where
\begin{align}
\label{eqn_P_hardening}
\Delta \bJ_{\cH}=\Lb \lb \sum_{(k,j) \in \cH}
\lb [\bP]_{kj}-\epsilon_{kj}  \rb \be_k^K {\be_j^N}^T \rb \otimes \bone_N   \Rb^T \bA^T.
\end{align}

\section{Monotonicity of $\phi(\cH)$}
\label{proof_lem_phi_monotone}
\begin{lemma}
	\label{lem_phi_monotone}
	$\phi(\varnothing)=0$ and $\phi(\cH)$ is a monotonic increasing set function.
\end{lemma}
\begin{proof}
	By definition $\phi(\varnothing)=0$ since $\Delta \bJ_{\varnothing}$ is a zero matrix. For any two sets $\cH_1 \subset \cH_2 \subset \cV_{app} \times \cV_{mac}$, 
	\begin{align}
	\label{eqn_lem_phi_monotone_1}
	\phi(\cH_2)-\phi(\cH_1)&= \by^T \lb  \Delta \bJ_{\cH_2}- \Delta \bJ_{\cH_1} \rb \by  \nonumber \\
	&= \by^T \lb  \Delta \bJ_{\cH_2 \setminus \cH_1} \rb \by  \nonumber \\
	& \geq 0 
	\end{align}
	since $\Delta \bJ_{\cH_2 \setminus \cH_1}$ is a nonnegative matrix and $\by$ is a nonnegative vector by the Perron-Frobenious theorem \cite{HornMatrixAnalysis}.
	Therefore, $\phi(\cH)$ is a monotonic increasing set function.
\end{proof}

\section{Efficient update of step 5 in Algorithm \ref{algo_greedy_edge_harden_score} with score recalculation}	
\label{proof_efficient_update}
Using the notations in Algorithm \ref{algo_greedy_edge_harden_score}, when hardening the edge $(k^*,j^*)$ the entry  $[\bP^{\eta}]_{k^*j^*}$ changes to $\epsilon_{k^*j^*}$. Let the original value of $[\bP^{\eta}]_{k^*j^*}$ before hardening be $\psi$. Then the update in step 5 is equivalent to 
\begin{align}
\bJ^{old}=\bJ^{old}-\bH^T \bA^T,
\end{align}
where $\bH=[\bzero, \ldots, \bh,\ldots,\bzero]$ is a matrix of zeros except that 
the  $[(k^*-1) \cdot N+1]$-th to $(k^* \cdot N)$-th entry of $\bH$'s $j^*$-th column $\bh$
is $\psi-\epsilon_{k^*j^*}$.

\section{Proof of Theorem \ref{thm_harden_wo_recal}}
\label{proof_thm_harden_wo_recal}
For any two hardening sets $\cH_1$ and $\cH_2$ satisfying $\cH_1 \subset \cH_2 \subset \cV_{app} \times \cV_{host}$, using (\ref{eqn_P_hardening}) and (\ref{eqn_lem_phi_monotone_1}) we have the additivity for the score function $\phi(\cH)$ as
\begin{align}
\label{eqn_thm_harden_wo_recal_1}
\phi(\cH_2)=\phi(\cH_1)+\phi(\cH_2 \setminus \cH_1).
\end{align}
For any hardening set $\cH$ of cardinality $|\cH|=\eta \geq 1$, let $\cH=\{ H_s \}_{s=1}^{\eta}$, where  $H_s$ is the $s$-th element in $\cH$, and let $\cH^{\eta}=\{ H^{\eta}_s\}_{s=1}^{\eta}$.
Then with (\ref{eqn_thm_harden_wo_recal_1}) we have
\begin{align}
\phi(\cH)&=\sum_{s=1}^{\eta} \phi(H_s)  \nonumber \\
& \leq \sum_{s=1}^{\eta} \phi(H^{\eta}_s)   \nonumber \\
& =\phi(\cH^{\eta}),
\end{align}		
where the maximum of $\phi(\cH)$ is attained when $\cH$ contains $\eta$ edges of highest hardening scores. Therefore, $\cH^{\eta}$ is a maximizer of $\phi(\cH)$.

\section{Proof of Theorem \ref{thm_hardening}}
\label{proof_thm_hardening}
We first show the relation that $\lambda_{\max} (\bJ) \geq \lambda_{\max} \lb \bJt (\cH) \rb$.	For any hardening set $\cH$, let $\byt$ be the largest eigenvector of $\bJt (\cH)$. With (\ref{eqn_P_hardening}) we have 
\begin{align}
\label{eqn_J_UB}
\lambda_{\max} \lb \bJt (\cH) \rb &= \byt^T \bJt \lb \cH \rb \byt  \nonumber \\
&=\byt^T \bJ  \byt - \byt^T \Delta \bJ_{\cH} \byt  \nonumber \\
& \leq \lambda_{\max} \lb \bJ  \rb - \byt^T \Delta \bJ_{\cH} \byt  \nonumber \\
& \leq \lambda_{\max} \lb \bJ  \rb.
\end{align}
The fact that $\byt^T \bJ  \byt  \leq \lambda_{\max} \lb \bJ  \rb$ is from the  Courant-Fischer theorem \cite{HornMatrixAnalysis}, and the last inequality uses the fact that $\byt^T \Delta \bJ_{\cH} \byt \geq 0$ from the  Perron-Frobenious theorem \cite{HornMatrixAnalysis} due to the fact that all entries in $\Delta \bJ_{\cH}$ and $\byt$ are nonnegative.

If $\lambda_{\max}(\bJ)>0$, then by 	(\ref{eqn_J_LB}) and (\ref{eqn_J_UB})
we have	$\phi(\cH^{opt})>0$. Otherwise $\phi(\cH^{opt})=0$ implies that $\by$ is a zero vector, which contradicts the fact that $\by$ 
is the largest eigenvector of $\bJ$. Therefore, if $\lambda_{\max} (\bJ)>0$ we have $\lambda_{\max} (\bJ) > \lambda_{\max} \lb \bJt (\cH^{opt}) \rb$.
When $|\cH|=\eta$, since $\cH^{opt}$ is the minimizer of  $\lambda_{\max} \lb \bJt (\cH) \rb $ and
$\cH^{\eta}$ is the maximizer of $\phi(\cH)$, we have
\begin{align}
\label{eqn_thm_hardening_1}
\lambda_{\max} \lb \bJt (\cH^\eta) \rb &\geq \lambda_{\max} \lb \bJt (\cH^{opt}) \rb  \nonumber \\
&\geq \lambda_{\max}(\bJ) -  \phi(\cH^{opt})  \nonumber \\
&\geq \lambda_{\max}(\bJ) -  \phi(\cH^{\eta}).
\end{align}
By the facts that $\lambda_{\max} (\bJ) > \lambda_{\max} \lb \bJt (\cH^{opt}) \rb$ and $\lambda_{\max} (\bJ) \geq \lambda_{\max} \lb \bJt (\cH^{\eta}) \rb$, if $\phi(\cH^{\eta})>0$, with (\ref{eqn_thm_hardening_1}) there exists some constant $c'' > 0 $ such that
\begin{align}
&\lambda_{\max}(\bJ)- c''  \cdot \phi(\cH^\eta)  \geq \lambda_{\max} \lb \bJt(\cH^{opt}) \rb;  \nonumber \\
&\lambda_{\max} \lb \bJt(\cH^{opt}) \rb \geq  \lambda_{\max}(\bJ)- \phi(\cH^\eta). 
\end{align}

\section{Proof of Corollary \ref{cor_hardening}}
\label{proof_cor_hardening}
This corollary is a direct result of Theorem \ref{thm_hardening} by replacing $\bJ$ with $\bJt(\cH)$ and setting $\eta=1$.

%


\section{Descriptions on the benchmark dataset}
\label{proof_bench}
This benchmark dataset was collected from the network traffic of a cyber testbed running inside a OpenStack-based cloud with nearly 2000 virtual machine instances.  Starting from a known machine (host), the attack involved logging from one machine to another using SSH.  Implemented by automated scripts, on each machine the attack replicated to four other machines at the beginning of every hour.  This process continued for 8 hours.  We collected network traffic flows from each virtual machine and combined to produce a 16 GB packet capture dataset.  Each packet information was further aggregated to produce ``flow" level information, which can be interpreted as a ``communication session" between two machines.  As an example, when a client connects to the server, the client may send 5 packets and receive 10 packets of data from the server.  The ``flow" level data will combine these 15 data packets into a single ``flow" to represent one interaction between the machines.   Each flow record has the following elements: IP address and port information for both source and destination devices, protocol, flow start time, duration and message size.  We infer the application by considering the protocol and destination address pair.  As an example, a flow to destination port 22 over TCP protocol implies an SSH connection.  To apply our proposed method to the cyber system against lateral movement attacks, we select the source and destination IP address, and the applications to build the host-application graph.

\end{document}

%% file: notation_20160706.tex
\newcommand{\Lb}{\left[}
\newcommand{\Rb}{\right]}
\newcommand{\lb}{\left(}
\newcommand{\rb}{\right)}

\newcommand{\col}{\textnormal{col}}

\newcommand{\bone}{\mathbf{1}}
\newcommand{\bzero}{\mathbf{0}}

\newcommand{\bd}{\mathbf{d}}

\newcommand{\bx}{\mathbf{x}}
\newcommand{\by}{\mathbf{y}}
\newcommand{\br}{\mathbf{r}}

\newcommand{\ba}{\mathbf{a}}
\newcommand{\bw}{\mathbf{w}}
\newcommand{\bh}{\mathbf{h}}

\newcommand{\bv}{\mathbf{v}}
\newcommand{\bfa}{\mathbf{a}}
\newcommand{\bat}{\tilde{\mathbf{a}}}
\newcommand{\be}{\mathbf{e}}
\newcommand{\bu}{\mathbf{u}}

\newcommand{\byt}{\widetilde{\mathbf{y}}}

\newcommand{\bP}{\mathbf{P}}
\newcommand{\bB}{\mathbf{B}}

\newcommand{\bX}{\mathbf{X}}

\newcommand{\bH}{\mathbf{H}}
\newcommand{\bI}{\mathbf{I}}
\newcommand{\bJ}{\mathbf{J}}

\newcommand{\bA}{\mathbf{A}}

\newcommand{\bW}{\mathbf{W}}
\newcommand{\bF}{\mathbf{F}}

\newcommand{\bAt}{\widetilde{\mathbf{A}}}
\newcommand{\bBt}{\widetilde{\mathbf{B}}}
\newcommand{\bPt}{\widetilde{\mathbf{P}}}
\newcommand{\bJt}{\widetilde{\mathbf{J}}}

\newcommand{\cX}{\mathcal{X}}
\newcommand{\cR}{\mathcal{R}}
\newcommand{\cH}{\mathcal{H}}

\newcommand{\cV}{\mathcal{V}}
\newcommand{\cE}{\mathcal{E}}
\newcommand{\cT}{\mathcal{T}}

\newcommand{\TB}{\mathbb{T}}
\newcommand{\HB}{\mathbb{H}}